\pdfoutput=1

\documentclass{article}

\usepackage[a4paper]{geometry}	
\usepackage{lmodern,microtype}	
\usepackage[english]{babel}		
\usepackage{xcolor}				
\usepackage{amsmath,amssymb,amsthm}
\usepackage{bm}
\usepackage{mathbbol}

\definecolor{darkblue}{rgb}{0,0,.8}
\definecolor{Rouge}{HTML}{C70039}

\usepackage[nobreak]{cite} 

\usepackage{hyperref}
\hypersetup{
    colorlinks,
    citecolor=red,
    filecolor=black,
    linkcolor=darkblue,
    urlcolor=black
}

\numberwithin{equation}{section} 

\usepackage[noabbrev,capitalize]{cleveref}

\theoremstyle{plain}
\newtheorem{theorem}{Theorem}[section]
\newtheorem{lemma}[theorem]{Lemma}
\newtheorem{proposition}[theorem]{Proposition}

\newtheorem{conjecture}[theorem]{Conjecture}

\newcommand{\chit}{\protect\raisebox{0.25ex}{$\chi$}}

\renewcommand{\i}{\mathrm{i}}
\newcommand{\ee}{\mathrm{e}}

\title{\Large\bf Sum rules for the supersymmetric eight-vertex model}
\author{\normalsize \textsc{Sandrine Brasseur} and \textsc{Christian Hagendorf}
\\
{\normalsize
  \begin{minipage}{\textwidth}
  \begin{center}
  \textit{
   Universit\'e catholique de Louvain\\
  Institut de Recherche en Math\'ematique et Physique\\
  Chemin du Cyclotron 2, 1348 Louvain-la-Neuve, Belgium} \\
  \bigskip
\href{mailto:sandrine.brasseur@uclouvain.be}{\normalsize 
\texttt{sandrine.brasseur@uclouvain.be}},
\href{mailto:christian.hagendorf@uclouvain.be}{\normalsize 
\texttt{christian.hagendorf@uclouvain.be}}
  \end{center}
  \end{minipage}
}
}

\date{}

\begin{document}
\maketitle

\begin{abstract}
The eight-vertex model on the square lattice with vertex weights $a,b,c,d$ obeying the relation $(a^2+ab)(b^2+ab)=(c^2+ab)(d^2+ab)$ is considered. Its transfer matrix with $L=2n+1,\, n\geqslant 0,$ vertical lines and periodic boundary conditions along the horizontal direction has the doubly-degenerate eigenvalue $\Theta_n = (a+b)^{2n+1}$. 
A basis of the corresponding eigenspace is investigated. 
Several scalar products involving the basis vectors are computed in terms of a family of polynomials introduced by Rosengren and Zinn-Justin. These scalar products are used to find explicit expressions for particular entries of the vectors. The proofs of these results are based on the generalisation of the eigenvalue problem for $\Theta_n$ to the inhomogeneous eight-vertex model.
\end{abstract}

\section{Introduction}

In the present work, we pursue an investigation of the square lattice eight-vertex model whose vertex weights $a,b,c,d$ obey the relation
\begin{equation}
\label{eqn:SUSY8VModel}
  (a^2+ab)(b^2+ab)=(c^2+ab)(d^2+ab).
\end{equation}
We consider the model on a cylinder with $L\geqslant 1$ vertical lines and periodic boundary conditions along the horizontal direction. If the number of vertical lines $L=2n+1,\, n\geqslant 0$, is odd, then the transfer matrix of the vertex model possesses the remarkably simple doubly-degenerate eigenvalue
\begin{equation}
\label{Def:HomEigenvalue}
  \Theta_n = (a+b)^{2n+1}.
\end{equation}
The existence of this eigenvalue was first conjectured by Stroganov \cite{stroganov:01}. Recently, Li\'enardy and Hagendorf proved this conjecture with the help of supersymmetry techniques \cite{hagendorf:12,hagendorf:18}.
Because of the supersymmetry, we refer to the case where the vertex weights obey \eqref{eqn:SUSY8VModel} as the \textit{supersymmetric eight-vertex model}. Several investigations, inspired by the simple form of $\Theta_n$, have revealed exciting connections between the supersymmetric eight-vertex model and a variety of topics in mathematical physics.
Examples include elliptic functional equations \cite{baxter:89,fabricius:05,rosengren:16}, the Painlev\'e VI equation \cite{bazhanov:05,bazhanov:06}, and special polynomials \cite{rosengren:13,rosengren:13_2,rosengren:14,rosengren:15}.

Our main objective is to investigate the eigenspace of $\Theta_n$. Bazhanov and Mangazeev \cite{mangazeev:10}, and Razumov and Stroganov \cite{razumov:10} initiated this investigation by exploiting the fact that the eigenvectors of $\Theta_n$ are the ground-state vectors of the Hamiltonian of a XYZ quantum spin chain.
They computed these vectors for small $n$ with the help of exact diagonalisation methods applied to the spin-chain Hamiltonian.
This computation led them to numerous conjectures on the eigenvectors, providing explicit expressions for their components and for scalar products. The latter are often referred to as \textit{sum rules}. Subsequently, a rigorous investigation of the eigenspace for general $n$ was undertaken by Zinn-Justin \cite{zinnjustin:13}. 
He considered a generalisation of the eigenvalue problem for $\Theta_n$ to the inhomogeneous eight-vertex model. Based on only one conjecture about its solution, he was able to derive one of the sum rules that had been observed by Bazhanov and Mangazeev. In this article, we use Zinn-Justin's conjecture to investigate the eigenspace of $\Theta_n$ in further detail. To this end, we apply a well-established induction technique of Izergin and Korepin's \cite{izergin:92,korepin:93}. It allows us to derive explicit expressions
for several scalar products of an eigenvector of the transfer matrix of the inhomogeneous eight-vertex model. From their homogeneous limits, we deduce several new sum rules and exact expressions for the components of the eigenvectors of $\Theta_n$.

The layout of this article is the following. 
In \cref{sec:8VModel}, we recall the definition of the transfer matrix of the inhomogeneous eight-vertex model. Focussing on the supersymmetric case, we discuss a generalisation of the eigenvalue problem for $\Theta_n$ to the inhomogeneous model and the properties of its solution that follow from Zinn-Justin's conjecture. In \cref{sec:ScalarProducts}, we introduce scalar products involving this solution as well as several solutions to the boundary Yang-Baxter equation. 
We obtain exact expressions for these scalar products in terms of the so-called elliptic Tsuchiya determinant.
In \cref{sec:HomogeneousLimit}, we compute the homogeneous limit of these expressions. 
We use them to derive and analyse several properties of the eigenvectors of $\Theta_n$. 
We present our conclusions in \cref{sec:Conclusion}.

\section{The inhomogeneous eight-vertex model}
\label{sec:8VModel}

In this section, we review known results on the eight-vertex model. We define the notations and conventions that we use throughout this article in \cref{sec:Notations}. In \cref{sec:TM8V}, we recall the definition of the transfer matrix of the inhomogeneous eight-vertex model, and some of its properties. From \cref{sec:SUSY8V} on, we focus on the supersymmetric eight-vertex model. We recall a generalisation of the eigenvalue problem for $\Theta_n$ as well as Zinn-Justin's conjecture on its solution. Moreover, we review several known properties of this solution.

\subsection{Notations}
\label{sec:Notations}
Throughout this article, $L$ denotes a positive integer. Let us consider the space $V^L=V_1\otimes \cdots \otimes V_L$, where $V_i=\mathbb C^2$ for each $i=1,\dots,L$. Its canonical basis vectors are labelled by sequences $\bm \alpha = \alpha_1\cdots\alpha_L$, with $\alpha_i\in\{\uparrow,\downarrow\}$ for each $i=1,\dots, L$. We refer to these sequences as spin configurations. The basis vector $|\bm \alpha\rangle$ corresponding to the spin configuration $\bm \alpha$ is 
\begin{equation}
\label{eqn:CanonicalBV}
  |\bm \alpha\rangle = |\alpha_1\rangle \otimes \cdots \otimes |\alpha_L\rangle,
\end{equation}
where
\begin{equation}
  |{\uparrow}\rangle = \begin{pmatrix} 1 \\ 0 \end{pmatrix},
  \quad
  |{\downarrow}\rangle = \begin{pmatrix} 0 \\ 1 \end{pmatrix}.
\end{equation}
Every vector $|\psi\rangle \in V^L$ can be expanded along the basis vectors:
\begin{equation}
  |\psi\rangle = \sum_{\bm \alpha} \psi_{\bm \alpha}|\bm \alpha\rangle.
\end{equation}
We refer to the numbers $\psi_{\bm \alpha}$ in this expansion as the components of $|\psi\rangle$. Furthermore, for every vector $|\phi\rangle \in V^L$, we define a dual vector by transposition $\langle \phi| = |\phi\rangle^t$. The corresponding dual pairing is 
\begin{equation}
  \label{eqn:DualPairing}
  \langle \phi|\psi\rangle = \sum_{\bm \alpha} \phi_{\bm \alpha}\psi_{\bm \alpha}.
\end{equation}
It defines a (real) scalar product on the subspace of $V^L$ of vectors with real components. We denote by $||\psi||^2 = \langle \psi|\psi\rangle$ the square norm of a vector in this subspace. For simplicity, we use the term `scalar product' and `square norm' even for vectors outside this subspace.

Let $M$ be an integer with $1\leqslant M \leqslant L$, and consider a linear
 operator $A\in \text{End}\,V^M$. For each choice of integers $1\leqslant i_1 < \cdots < i_M \leqslant L$, we denote by $A_{i_1,\dots, i_M}$ the canonical embedding of $A$ into $\text{End}\,V^L$.  It acts nontrivially only on the factors $V_{i_1},\dots,V_{i_M}$ of $V^L$, and acts like the identity operator on all other factors. As an example, we consider the case where $M=1$ and $A$ is one of the Pauli matrices 
\begin{equation}
  \sigma^x = \begin{pmatrix} 0 & 1 \\ 1 & 0 \end{pmatrix},
  \quad
  \sigma^y = \begin{pmatrix} 0 & -\i \\ \i & 0 \end{pmatrix},
  \quad 
  \sigma^z = \begin{pmatrix} 1 & 0 \\ 0 & -1 \end{pmatrix}.
\end{equation}
On $V^L$, we have
\begin{equation}
  \sigma_i^\kappa= \underbrace{\mathbb 1\otimes \cdots \otimes \mathbb 1}_{i-1} \otimes \,\sigma^\kappa \otimes \underbrace{\mathbb 1 \otimes \cdots \otimes \mathbb 1}_{L-i}, \quad \kappa = x,y,z,
\end{equation}
for each $i=1,\dots,L$, where $\mathbb 1$ denotes the $2\times 2$ identity matrix.
In terms of these operators, we define the spin-parity operator $P$ and the spin-reversal operator $F$ as
\begin{equation}
  \label{eqn:DefP}
  P = (-1)^L\prod_{i=1}^L \sigma_i^z, \qquad  F = \prod_{i=1}^L\sigma_i^x.  \end{equation}
They obey the relation
\begin{equation}
  \label{eqn:FPCR}
  FP = (-1)^{L}PF.
\end{equation}

\subsection{The transfer matrix}
\label{sec:TM8V}

\subsubsection{The \texorpdfstring{$\boldsymbol{R}$}{R}-matrix}

The eight-vertex model is an integrable vertex model of planar statistical mechanics that can be constructed from an $R$-matrix \cite{baxterbook}. This $R$-matrix is an operator $R \in \text{End}\,V^2$. In the basis $\{|{\uparrow\uparrow}\rangle, |{\uparrow\downarrow}\rangle, |{\downarrow\uparrow}\rangle,|{\downarrow\downarrow}\rangle\}$ it acts like the matrix
\begin{equation}
   R =
   \begin{pmatrix}
   a & 0 & 0 & d\\
   0 & b & c & 0\\
   0 & c & b & 0\\
   d & 0 & 0 & a
   \end{pmatrix}.
\end{equation}
Here, the numbers $a,b,c,d$ are the vertex weights. It is often convenient to use a parameterisation of these weights in terms of the Jacobi theta functions $\vartheta_i(u) = \vartheta_i(u,p),\,i=1,\dots,4$.\footnote{We use the classical notation of Whittaker and Watson for the Jacobi theta functions. Throughout this article, we often omit the intermediate steps of our calculations that involve the many identities between these functions \cite{whittaker:27,gradshteyn:07}.} In our parameterisation, the weights depend on a spectral parameter $u$, the crossing parameter $\eta$ and the elliptic nome $p=\ee^{\i \pi \tau},\,\text{Im}\,\tau > 0$. Up to an irrelevant overall factor, they are given by 
\begin{align}
  \label{eqn:WeightsTheta}
  \begin{split}
  a(u) = \vartheta_4(2\eta,p^2)\vartheta_1(u+2\eta,p^2)\vartheta_4(u,p^2),\\
  b(u) = \vartheta_4(2\eta,p^2)\vartheta_4(u+2\eta,p^2)\vartheta_1(u,p^2),\\
  c(u) = \vartheta_1(2\eta,p^2)\vartheta_4(u+2\eta,p^2)\vartheta_4(u,p^2),\\
  d(u) = \vartheta_1(2\eta,p^2)\vartheta_1(u+2\eta,p^2)\vartheta_1(u,p^2).
  \end{split}
\end{align}%
Here and in the following, we only write out the dependence on the spectral parameter $u$, but not on $\eta$ and $p$. The parameterisation implies that the $R$-matrix obeys the so-called Yang-Baxter equation. On $V^3$, it is given by
\begin{equation}
  \label{eqn:YBE}
  R_{1,2}(u-v)R_{1,3}(u)R_{2,3}(v)=  R_{2,3}(v)R_{1,3}(u)R_{1,2}(u-v).
\end{equation}
Moreover, we note that $R(0) =\vartheta_4(0,p^2) \vartheta_1(2\eta,p^2)\vartheta_4(2\eta,p^2)\mathcal P$, where $\mathcal P$ is the permutation operator on $V^2$, acting on the basis vectors as $\mathcal P|\alpha_1\alpha_2\rangle = |\alpha_2\alpha_1\rangle$. 
We use it to define the $\check R$-matrix $\check R(u) = \mathcal P R(u)$. It obeys the braid form of the Yang-Baxter equation
\begin{equation}
  \label{eqn:BraidYBE}
  \check R_{1,2}(u-v)\check R_{2,3}(u)\check R_{1,2}(v)=  \check R_{2,3}(v)\check R_{1,2}(u)\check R_{2,3}(u-v).
\end{equation}

\subsubsection{The transfer matrix}

We now consider a square lattice with $L$ vertical lines and periodic boundary conditions along the horizontal direction. The transfer matrix of the eight-vertex model on this lattice is 
\begin{equation}
  \label{eqn:TM8V}
  T(u|u_1,\dots,u_L) = \text{tr}_0 (R_{0,L}(u_L-u)\cdots R_{0,2}(u_2-u)R_{0,1}(u_1-u)).
\end{equation}
In this expression, the $R$-matrices act on the tensor product $V_0\otimes V^L$, where $V_0=\mathbb C^2$ is called the auxiliary space. The trace is taken over this auxiliary space. We refer to $u$ as the spectral parameter and $u_1,\dots,u_L$ as the inhomogeneity parameters. If $u_1=\cdots = u_L = 0$, then we call the eight-vertex model described by this transfer matrix \textit{homogeneous}, otherwise \textit{inhomogeneous}.

The transfer matrix is diagonalisable, at least for real $u,\eta,p$ and real $u_1,\dots,u_L$ sufficiently close to $0$ \cite{baxterbook}. It is possible to choose its eigenvectors to be independent of the spectral parameter $u$. This property follows from the commutation relation 
\begin{equation}
  [T(u|u_1,\dots,u_L),T(v|u_1,\dots,u_L)] = 0,
\end{equation}
for all $u,v$, which itself is a consequence of the Yang-Baxter equation \eqref{eqn:YBE} \cite{baxterbook}. Moreover, the transfer matrix possesses two elementary symmetries: It is invariant under spin reversal and preserves the spin parity. These symmetries lead to the commutation relations
\begin{align}
    [T(u|u_1,\dots,u_L),F]=[T(u|u_1,\dots,u_L),P]=0.  
    \label{eqn:TFPCR} 
\end{align}
We note that for odd $L$, the relations \eqref{eqn:FPCR} and \eqref{eqn:TFPCR} imply that each transfer-matrix eigenvalue has an even degeneracy.

\subsection{The supersymmetric point}
\label{sec:SUSY8V}

Unless stated otherwise, we consider throughout the remainder of this article the case where $L=2n+1$, with $n\geqslant 0$ an integer, and
\begin{equation}
  \eta = \pi/3.
\end{equation} 
For this value of the crossing parameter, the weights \eqref{eqn:WeightsTheta} obey the defining relation \eqref{eqn:SUSY8VModel} of the supersymmetric eight-vertex model. Moreover, the transfer matrix of the inhomogeneous model is conjectured \cite{razumov:10} to possess the doubly--degenerate eigenvalue
\begin{equation}
  \Theta_n(u|u_1,\dots,u_{2n+1}) = \prod_{i=1}^{2n+1}r(u_i-u),
\end{equation}
where
\begin{equation}
  r(u) = a(u)+b(u) = \vartheta_4(0,p^2)\vartheta_1(u+\eta,p^2)\vartheta_4(u+\eta,p^2).
  \label{eqn:DefR}
\end{equation}
For the homogeneous model, where $u_1=\dots=u_{2n+1}=0$, it reduces to the
eigenvalue $\Theta_n=(a+b)^{2n+1}$, whose existence has been rigorously established \cite{hagendorf:18}.

\subsubsection{The eigenvalue problem}

Since the transfer-matrix eigenvalue is doubly degenerate, it is convenient to lift the degeneracy by imposing the transformation property of the eigenvectors under spin reversal. We consider the following eigenvalue problem
\begin{align}
  \label{eqn:EVProblemInhom}
  T(u|u_1,\dots,u_{2n+1})|\Psi\rangle  =\Theta_n(u|u_1,\dots,u_{2n+1})|\Psi\rangle, \quad 
  F|\Psi\rangle = (-1)^n|\Psi\rangle,
\end{align}
where $|\Psi\rangle \in V^{2n+1}$. For small $n$, it is possible to find nontrivial solutions to this system of linear equations through a direct calculation. For arbitrary $n$, however, their existence has, to our best knowledge, not been rigorously established. The following conjecture postulates this existence and characterises the corresponding space of solutions.
\begin{conjecture}[Zinn-Justin \cite{zinnjustin:13}]
  \label{conj:PZJ}
  The space of solutions to the eigenvalue problem \eqref{eqn:EVProblemInhom} is spanned by a vector
\begin{equation}
    |\Psi_n\rangle = |\Psi_n(u_1,\dots,u_{2n+1})\rangle \in V^{2n+1},
\end{equation}
whose components are entire functions with respect to $u_i$, for each $i=1,\dots,2n+1$. They are generically non-vanishing and do not have a common factor that depends on $u_1,\dots,u_{2n+1}$. Furthermore, the vector obeys the relations
\begin{align}
  |\Psi_n(\dots,u_i+2\pi \tau,\dots)\rangle &= p^{-4n}\ee^{-2\i\sum_{j=1}^{2n+1}(u_i-u_j)}|\Psi_n(\dots,u_i,\dots)\rangle,\\
  |\Psi_n(\dots,u_i+\pi,\dots)\rangle &= \prod_{j=1,j\neq i}^{2n+1}\sigma_j^z|\Psi_n(\dots,u_i,\dots)\rangle.
\end{align}
\end{conjecture}
In this article, we assume that this conjecture holds and derive our main results from it.

\subsubsection{Properties of the eigenvector}

\Cref{conj:PZJ} implies several properties of the vector $|\Psi_n\rangle$,  which we frequently use in the following. We recall them here in the form of four lemmas, whose proofs can be found in \cite{zinnjustin:13}. The first lemma expresses the action of the $\check R$-matrix on the vector.
\begin{lemma}[Exchange relations]
\label{lem:Exchange}
  For $n\geqslant 1$ and each $i=1,\dots,2n$, we have
\begin{equation}
    \check R_{i,i+1}(u_{i+1}-u_i)|\Psi_n(\dots,u_i,u_{i+1},\dots)\rangle = r(u_{i+1}-u_i)|\Psi_n(\dots,u_{i+1},u_i,\dots)\rangle.
\end{equation}
\end{lemma}
We note that the compatibility of the exchange relations follows from the braid Yang-Baxter equation \eqref{eqn:BraidYBE}.

The second lemma describes the action of a local spin-flip operator on the eigenvector.
\begin{lemma}[Local spin flips]
\label{lem:SRPsi}
For each $i=1,\dots,2n+1$, we have
\begin{equation}
  \sigma_i^x|\Psi_n(\dots,u_i,\dots)\rangle = (-p)^n\ee^{-\i\sum_{j=1}^{2n+1}(u_j-u_i)}|\Psi_n(\dots,u_i+\pi \tau,\dots)\rangle.
\end{equation}
\end{lemma}

For the third lemma, we consider three inhomogeneity parameters $u_i,u_j,u_k$ with $1\leqslant i < j< k\leqslant 2n+1$. We say that these parameters form a \textit{wheel} if they can be written as $u_i=x,\,u_j=x+2\eta,\,u_{k}=x+4\eta$, for some $x$.
The so-called \textit{wheel condition} asserts that the vector vanishes whenever a wheel is formed.
\begin{lemma}[Wheel condition]
For $n\geqslant 1$ and all $x$, we have
\label{lem:WheelCondition}
\begin{equation}
  |\Psi_n(\dots,x,\dots,x+2\eta,\dots, x+4\eta,\dots)\rangle = 0.
\end{equation}
\end{lemma}

The fourth lemma provides so-called \textit{reduction relations} for the vector: Upon a specialisation of certain inhomogeneity parameters, $|\Psi_n\rangle$ can be expressed in terms of $|\Psi_{n-1}\rangle$. To state the reduction relations, we define for each $L\geqslant 1$ and each $1\leqslant i \leqslant L+1$ a linear operator $\varphi_i:V^L \to V^{L+2}$. It acts on the basis vector labelled by $\bm \alpha = \alpha_1\cdots \alpha_L$ according to 
\begin{equation}
  \varphi_i|\bm \alpha\rangle =
    \bigotimes_{j=1}^{i-1}|\alpha_j\rangle \otimes |s\rangle \otimes \bigotimes_{j=i}^{L}|\alpha_j\rangle,
    \end{equation}
where 
 \begin{equation}
  \label{eqn:Defs}
  |s\rangle = |{\uparrow\downarrow}\rangle-|{\downarrow\uparrow}\rangle.
\end{equation}
\begin{lemma}[Reduction relations]
\label{lem:Reduction}
  For $n\geqslant 1$ and each $i=1,\dots,2n$, we have
  \begin{multline}
  |\Psi_n(\dots,u_{i-1},u_i,u_{i+1} 
  = u_i+2\eta,u_{i+2}\dots)\rangle \\= C_n \left(\prod_{\substack{j=1,\,j\neq i,i+1}}^{2n+1}\vartheta_1(u_i-u_j-2\eta)\right)\varphi_i|\Psi_{n-1}(\dots,u_{i-1},u_{i+2},\dots)\rangle,
\end{multline}
where $C_n$ is independent of $u_1,\dots,u_{2n+1}$.
\end{lemma}

\subsubsection{Normalisation}

The constant $C_n$ in \cref{lem:Reduction} depends on the relative normalisation of the vectors. It is not fixed by \cref{conj:PZJ}. In the following, we choose
\begin{equation}
  \label{eqn:ChoiceForC}
  C_n = 1,
\end{equation}
for each $n\geqslant 1$.  Furthermore, we define
\begin{equation}
  \label{eqn:Psi0}
  |\Psi_0\rangle = |{\uparrow}\rangle + |{\downarrow}\rangle.
\end{equation}
These choices and the properties listed in \cref{conj:PZJ} completely fix the vector $|\Psi_n\rangle$ for each $n\geqslant 0$ \cite{zinnjustin:13}. For instance, for $n=1$ its components are
\begin{align}
\label{eqn:Psi1}
\begin{split}
  (\Psi_1)_{\uparrow\uparrow\uparrow}=-(\Psi_1)_{\downarrow\downarrow\downarrow}  &= \rho\, \vartheta_1(u_2-u_1+\eta,p^2)\vartheta_1(u_3-u_2+\eta,p^2)\vartheta_1(u_1-u_3+\eta,p^2),\\ 
  (\Psi_1)_{\uparrow\downarrow\downarrow}=-(\Psi_1)_{\downarrow\uparrow\uparrow} &= \rho\, \vartheta_4(u_2-u_1+\eta,p^2)\vartheta_1(u_3-u_2+\eta,p^2)\vartheta_4(u_1-u_3+\eta,p^2),\\
  (\Psi_1)_{\downarrow\uparrow\downarrow}=- (\Psi_1)_{\uparrow\downarrow\uparrow} &= \rho\, \vartheta_4(u_2-u_1+\eta,p^2)\vartheta_4(u_3-u_2+\eta,p^2)\vartheta_1(u_1-u_3+\eta,p^2),\\
  (\Psi_1)_{\downarrow\downarrow\uparrow}= -(\Psi_1)_{\uparrow\uparrow\downarrow}&= \rho\, \vartheta_1(u_2-u_1+\eta,p^2)\vartheta_4(u_3-u_2+\eta,p^2)\vartheta_4(u_1-u_3+\eta,p^2),
\end{split}
\end{align}
where the overall factor is given by
\begin{equation}
  \rho = \frac{2}{\vartheta_2(0)\vartheta_4(0,p^2)}.
  \label{eqn:NormFactor}
\end{equation}
For $n\geqslant 2$, the components are considerably more complicated. An explicit formula for these components remains to be found.

\section{Scalar products}
\label{sec:ScalarProducts}

In this section, we compute several scalar products that involve the transfer-matrix eigenvector $|\Psi_n\rangle$. In \cref{sec:DefinitionSP}, we define the scalar products with the help of two solutions to the boundary Yang-Baxter equation. 
In \cref{sec:PropertiesSP}, we analyse their symmetry and analyticity properties, and show that they obey a set of reduction relations.
These properties are sufficient to characterise the scalar products uniquely.
In \cref{sec:Determinants}, we find determinant formulas in terms of the so-called elliptic Tsuchiya determinant that fulfil the same properties as the scalar products. Using uniqueness and an inductive proof technique, we show that they provide explicit expressions for the scalar products.
In \cref{sec:RGPolynomials}, we recall a relation between Tsuchiya determinants and a family of multivariable polynomials introduced by Zinn-Justin and Rosengren.

\subsection{Definition}
\label{sec:DefinitionSP}

In its vector form, the boundary Yang-Baxter equation for the eight-vertex model is given by
\begin{equation}
\label{eqn:bYBE}
  \check R_{1,2}(x-y)\check R_{2,3}(-x-y)|\chi(x)\rangle\otimes |\chi(y)\rangle = \check R_{3,4}(x-y)\check R_{2,3}(-x-y)|\chi(y)\rangle\otimes |\chi(x)\rangle,
\end{equation}
where $|\chi(x)\rangle \in V^2$.\footnote{Traditionally, the boundary Yang-Baxter equation is presented in a matrix form. Its solutions are the so-called $K$-matrices. To each $K$-matrix $K(x)$ corresponds a solution $|\chi(x)\rangle$ of the vector form.} We consider the following two solutions of this equation:
\begin{align}
  \label{eqn:DefChi}
  \begin{split}
  |\chit(x)\rangle &= \vartheta_1(x+\lambda,p^2)\vartheta_4(x-\lambda-2\eta,p^2)|{\uparrow\downarrow}\rangle+\vartheta_1(x-\lambda-2\eta,p^2)\vartheta_4(x+\lambda,p^2)|{\downarrow\uparrow}\rangle,\\
  |\bar\chit(x)\rangle & = \vartheta_1(x+\lambda,p^2)\vartheta_1(x-\lambda-2\eta,p^2)|{\uparrow\uparrow}\rangle+\vartheta_4(x-\lambda-2\eta,p^2)\vartheta_4(x+\lambda,p^2)|{\downarrow\downarrow}\rangle.
  \end{split}
\end{align}
Here, $\lambda$ is an arbitrary parameter. Both these solutions are specialisations of the most general solution to the boundary Yang-Baxter equation for the eight-vertex model found by Inami and Konno \cite{inami:94}. In the following two lemmas, we list some useful properties of $|\chit(x)\rangle$ and $|\bar \chit(x)\rangle$. They follow from standard identities for Jacobi theta functions. It will be convenient to define
\begin{equation}
 g(x)=\frac{\vartheta_4(2(\eta+x),p^2)}{\vartheta_4(2(\eta-x),p^2)}, \quad \bar g(x) =\frac{\vartheta_1(2(\eta+x),p^2)}{\vartheta_1(2(\eta-x),p^2)}.
  \label{eqn:DefG}
\end{equation}

\begin{lemma}
\label{lem:Fish}
We have the relations
  \begin{equation}
  \check R(2x)|\chit(x)\rangle = g(x)r(2x)|\chit(-x)\rangle, \quad \check R(2x)|\bar \chit(x)\rangle = \bar g(x)r(2x)|\bar\chit(-x)\rangle.
\end{equation}
\end{lemma}
\begin{lemma}
\label{lem:SpecialSP}
We have the matrix elements
\begin{align}
  \frac{\left(\langle \chi(x+2\eta)|\otimes \langle \chi(x)|\right)\check R_{2,3}(-2(x+\eta))\left(|s\rangle \otimes |s\rangle\right)}{r(-2(x+\eta))} & = \frac{\vartheta_2(0)^2\vartheta_1(x-\lambda)\vartheta_1(x+\lambda+2\eta) g(x)}{2},\\
   \frac{\left(\langle \bar \chi(x+2\eta)|\otimes \langle \bar \chi(x)|\right)\check R_{2,3}(-2(x+\eta))\left(|s\rangle \otimes |s\rangle\right)}{r(-2(x+\eta))} & = \frac{\vartheta_2(0)^2\vartheta_1(x-\lambda)\vartheta_1(x+\lambda+2\eta) \bar g(x)}{2},
\end{align}
where $|s\rangle$ is the vector defined in \eqref{eqn:Defs}.
\end{lemma}

For each $n\geqslant 1$, we use the two solutions of the boundary Yang-Baxter equation to define the vectors 
\begin{align}
  |\xi_n(x_1,\dots,x_n)\rangle &= \Bigg(\bigotimes_{i=1}^n |\chit(x_i)\rangle \Bigg)\otimes |{\uparrow}\rangle,\\
  |\bar\xi^\pm_n(x_1,\dots,x_n)\rangle &= \Bigg(\bigotimes_{i=1}^n |\bar\chit(x_i)\rangle \Bigg)\otimes \left(|{\uparrow}\rangle\pm |{\downarrow}\rangle\right).
\end{align}
It will also be convenient to introduce
\begin{equation}
  |\xi_0\rangle =|{\uparrow}\rangle,\quad|\bar\xi_0^\pm\rangle = |{\uparrow}\rangle\pm|{\downarrow}\rangle.
\end{equation}
Using these vectors, we define the scalar products
\begin{align}
\label{eqn:DefZ}
 Z_n(x_1,\dots,x_n) &= \langle \xi_n(x_1,\dots,x_n)|\Psi_n(x_1,-x_1,\dots,x_n,-x_n,0)\rangle,\\
 \label{eqn:DefZBar}
 \bar Z_n^\pm(x_1,\dots,x_n) &= \langle \bar\xi^\pm_n(x_1,\dots,x_n)|\Psi_n(x_1,-x_1,\dots,x_n,-x_n,0)\rangle.
\end{align}
In the following, we only write out their dependence on the inhomogeneity parameters $x_1,\dots,x_n$ if necessary.

It is straightforward to find $Z_n$ and $\bar Z_n^\pm$ for $n=0$ and $n=1$. Using \eqref{eqn:Psi0}, we obtain
\begin{equation}
  Z_0=1,\quad \bar Z_0^+ = 2, \quad \bar Z_0^-=0.
\end{equation}
Furthermore, we use the components \eqref{eqn:Psi1} and find
\begin{align}
\label{eqn:Zn1}
\begin{split}
  Z_1 &= \frac{1}{2}\rho\, \vartheta_2(0)\vartheta_2(\eta+\lambda)\vartheta_4(2(\eta+x_1),p^2)\vartheta_1(\eta-x_1)\vartheta_1(\eta+x_1),\\
    \bar Z_1^+ &= \rho\, \vartheta_4(0) \vartheta_4(\eta+\lambda)\vartheta_1(2(\eta+x_1),p^2)\vartheta_3(\eta+x_1)\vartheta_3(\eta-x_1),\\
  \bar Z_1^- &= \rho\, \vartheta_3(0) \vartheta_4(\eta+\lambda)\vartheta_1(2(\eta+x_1),p^2)\vartheta_4(\eta+x_1)\vartheta_4(\eta-x_1),
  \end{split}
\end{align}
where $\rho$ is defined in \eqref{eqn:NormFactor}.

\subsection{Properties}
\label{sec:PropertiesSP}

In this section, we establish the properties of $Z_n$ and $\bar Z_n^\pm$ that allow us to find them for $n\geqslant 2$. The proofs of these properties are very similar for $Z_n$ and $\bar Z_n^\pm$. Hence, we focus on $Z_n$. If necessary, we indicate the modifications to be made for $\bar Z_n^\pm$.

\subsubsection{Symmetry}

\begin{proposition}
  \label{prop:SymmetryZ}
  For $n\geqslant 2$, $Z_n$ and $\bar Z_n^\pm$ are symmetric functions of $x_1,\dots,x_n$.
\end{proposition}
\begin{proof}
    We sketch the proof of the symmetry of $Z_n$. It is sufficient to prove that 
    \begin{equation}
      Z_n(\dots,x_i,x_{i+1},\dots) = Z_n(\dots,x_{i+1},x_{i},\dots),
      \label{eqn:ZTransposition}
    \end{equation}
    for each $i=1,\dots,n-1$. For $i=1$, we find
    \begin{align}
       Z_n(x_1,x_2,\dots) &= \langle \xi_{n}(x_1,x_2,\dots)|\Psi_n(x_1,-x_1,x_2,-x_2,\dots)\rangle \nonumber\\
      & =\frac{\langle \xi_{n}(x_1,x_2,\dots)|\check R_{2,3}(-x_2-x_1)\check R_{3,4}(x_2-x_1)|\Psi_n(x_1,x_2,-x_2,-x_1,\dots)\rangle}{r(-x_2-x_1)r(x_2-x_1)}\nonumber\\
      & =\frac{\langle \xi_{n}(x_2,x_1,\dots)|\check R_{2,3}(-x_2-x_1)\check R_{1,2}(x_2-x_1)|\Psi_n(x_1,x_2,-x_2,-x_1,\dots)\rangle}{r(-x_2-x_1)r(x_2-x_1)} \nonumber\\
     & = \langle \xi_{n}(x_2,x_1,\dots)|\Psi_n(x_2,-x_2,x_1,-x_1,\dots)\rangle = Z_n(x_2,x_1,\dots).
    \end{align}
From the first to the second line, we used \cref{lem:Exchange}. The third line is obtained by using the symmetry of the $\check R$-matrix and the boundary Yang-Baxter equation \eqref{eqn:bYBE}. The fourth line is the result of another application of \cref{lem:Exchange}. This establishes \eqref{eqn:ZTransposition} for $i=1$. The cases where $i=2,\dots,n-1$ are straightforward generalisations.
\end{proof}

\begin{proposition}
  \label{prop:InvZ}
  For each $i=1,\dots,n$, we have
   \begin{align}
  \vartheta_4(2(\eta+x_i),p^2)Z_n(\dots,-x_i,\dots) &=\vartheta_4(2(\eta-x_i),p^2)Z_n(\dots,x_i,\dots),\\
  \vartheta_1(2(\eta+x_i),p^2)\bar Z_n^\pm(\dots,-x_i,\dots)& =\vartheta_1(2(\eta-x_i),p^2)\bar Z_n^\pm(\dots,x_i,\dots).
  \end{align}
\end{proposition}
  \begin{proof}
  We present the proof of the first relation. By \cref{prop:SymmetryZ}, it is sufficient to consider $i=1$. We compute
    \begin{multline}
Z_n(-x_1,\dots) = \langle \xi_{n}(-x_1,\dots)|\Psi_n(-x_1,x_1,\dots)\rangle\\
=  \frac{\langle \xi_{n}(x_1,\dots)|\check R_{1,2}(2x_1)|\Psi_n(-x_1,x_1,\dots)\rangle}{g(x_1)r(2x_1)}= g(-x_1) Z_n(x_1,\dots).
  \end{multline}
  From the first to the second line, we used \cref{lem:Fish} and the symmetry of the $\check R$-matrix. The equality in the second line follows from \cref{lem:Exchange} and the explicit expression of $g(x)$, given above. 
  
    The proof of the second relation is similar.
\end{proof}

\subsubsection{Analyticity}

In the following, we use the concept of a \textit{theta function} \cite{rosengren:09}. Let $m\geqslant 0$ be an integer. A theta function of degree $m$, nome $p$ and norm $t$ is an entire function $f$ with the pseudo-periodicity properties 
\begin{equation}
  f(z+\pi) = (-1)^m f(z), \quad f(z+\pi\tau) = (-p)^{-m}\ee^{-2\i(mz-t)}f(z).
\end{equation}
Clearly, the norm $t$ is only defined modulo $\pi$. The resulting ambiguity is, however, not important for our considerations.
A simple (nontrivial) example of a theta function of degree $m=1$, nome $t$ and norm $p$ is $f(z) = \vartheta_1(z-t)$.
\begin{proposition}
  \label{prop:PerZ}
  For $n\geqslant 1$, $Z_n$ and $\bar Z_n^\pm$ are theta functions of degree $2(n+1)$, nome $p$ and norm $\frac{\pi}{2}+\eta$ with respect to $x_i$ for each $i=1,\dots,n$.
  \begin{proof}
  We focus on $Z_n$. By \cref{prop:SymmetryZ}, it is sufficient to establish the statement for $Z_n$ as a function of $x_1$. It follows from \cref{conj:PZJ} and from the definition of the vector $|\chi(x)\rangle$ that $Z_n$ is an entire function of $x_1$. To establish its pseudo-periodicity properties, we note that $|\chit(x)\rangle$ obeys the relations
\begin{equation}
\label{eqn:PerChi}
 |\chit(x+\pi)\rangle  = \sigma_1^z\sigma_1^z|\chit(x)\rangle, \quad 
  |\chit(x+\pi \tau)\rangle  =-p^{-1}\ee^{-2\i(x-\eta)}\sigma_1^x \sigma_2^x|\chit(x)\rangle .
\end{equation}
  Furthermore, it follows from \cref{conj:PZJ} and \cref{lem:SRPsi} that
 \begin{equation}
     \label{eqn:PsiShifts}
    \begin{aligned}
      |\Psi_n(x_1+\pi,-x_1-\pi,\dots)\rangle &= \sigma_1^z\sigma_2^z|\Psi_n(x_1,-x_1,\dots)\rangle,\\
       |\Psi_n(x_1+\pi\tau,-x_1-\pi \tau,\dots)\rangle &= p^{-(2n+1)}\ee^{-2(2n+1)\i x_1}\sigma_1^x\sigma_2^x|\Psi_n(x_1,-x_1,\dots)\rangle.
    \end{aligned}
    \end{equation}
    We combine \eqref{eqn:PerChi} and \eqref{eqn:PsiShifts} to obtain
    \begin{align}
      Z_n(x_1+\pi,\dots) &= \langle \xi_{n}(x_1+\pi,\dots)|\sigma_1^z\sigma_2^z|\Psi_n(x_1,-x_1,\dots)\rangle=Z_n(x_1,\dots),
    \end{align}
     and
     \begin{align}
    Z_n(x_1+\pi\tau,\dots) &=p^{-(2n+1)}\ee^{-2(2n+1)\i x_1}\langle \xi_{n}(x_1+\pi\tau,\dots)|\sigma_1^x\sigma_2^x|\Psi_n(x_1,-x_1,\dots)\rangle\\
    &=p^{-2(n+1)}\ee^{-2\i(2(n+1)x_1-(\eta+\pi/2))}Z_n(x_1,\dots) .
  \end{align}
  This ends the proof for $Z_n$. The proof for $\bar Z_n^\pm$ is similar.
  \end{proof}
 \end{proposition}

\subsubsection{Zeroes and trivial factors}

In the next proposition, we identify certain trivial zeroes of the scalar products $Z_n$ and $\bar Z_n^\pm$. It is practical to introduce the abbreviations
\begin{equation}
\label{eqn:DefBeta}
 \beta_1=\eta,\quad \beta_2 = \eta+\frac{\pi}{2}, \quad \beta_3= \eta+\frac{\pi}{2}+\frac{\pi\tau}{2}, \quad \beta_4=\eta+\frac{\pi \tau}{2}.
\end{equation}

\begin{proposition}
\label{prop:TrivialZeroes}
 For $n\geqslant 1$ and each $i=1,\dots,n$, we have
\begin{align}
  Z_n(\dots,x_i  = -\beta_1, \dots)&=Z_n(\dots,x_i = \beta_1, \dots) =  0,
  \label{eqn:ZZeroBeta1}\\
  \label{eqn:ZZeroBeta34}
  Z_n(\dots, x_i=-\beta_3,\dots) & = Z_n(\dots, x_i=-\beta_4,\dots)= 0,\\
  \label{eqn:ZBarZeroBeta12}
  \bar Z_n^{\pm}(\dots, x_i=-\beta_1,\dots) &= \bar Z_n^\pm(\dots, x_i=-\beta_2,\dots)= 0.
\end{align}
\end{proposition}
\begin{proof}
  By \cref{prop:SymmetryZ}, it is sufficient to prove the proposition for $i=1$. We focus on the proof of \eqref{eqn:ZZeroBeta1}. We note that
    \begin{align}
      Z_n(x_1=-\beta_1,\dots)&=\langle \xi_{n}(\eta,x_2,\dots,x_n)|\Psi_n(-\eta,\eta,x_2,-x_2,\dots,x_n,-x_n,0)\rangle\\
      &=\langle \xi_{n}(\eta,x_2,\dots,x_n)|\prod_{j=1}^{2n}\sigma_j^z|\Psi_n(-\eta,\eta,x_2,-x_2,\dots,x_n,-x_n,3\eta)\rangle .
      \nonumber
    \end{align}
    The first, second and last arguments of $|\Psi_n\rangle$ form a wheel. \Cref{lem:WheelCondition} implies that the vector vanishes identically. Hence, $Z_n(x_1=-\beta_1,\dots)=0$. Furthermore, it follows from \cref{prop:InvZ} that
    \begin{equation}
      \vartheta_4(4\eta,p^2)Z_n(x_1=-\beta_1,\dots)= \vartheta_4(0,p^2)Z_n(x_1=\beta_1,\dots).
    \end{equation}
    Since $\vartheta_4(0,p^2)$ is non-zero, we find that $Z_n(x_1=\beta_1,\dots)=0$. Hence, we obtain \eqref{eqn:ZZeroBeta1} with $i=1$.    
         
    Finally, we note that \eqref{eqn:ZZeroBeta34} and \eqref{eqn:ZBarZeroBeta12} directly follow from \cref{prop:InvZ} and the known zeroes of the Jacobi theta functions.
\end{proof}

We now use the knowledge of the zeroes to find trivial factors of the scalar products. To this end, we use a factorisation property of theta functions. It follows from standard complex analysis \cite{ahlfors:79,filali:11}:
\begin{lemma}
  \label{lem:Factorisation}
  Let $f$ be a theta function of degree $m\geqslant 1$, nome $p$ and norm $t$. Let $\xi$ be a complex number such that $f(\xi)=0$, then there exists a theta function $g$ of degree $m-1$, nome $p$ and norm $t-\xi$ such that
  \begin{equation}
    f(z) = \vartheta_1(z-\xi)g(z).
  \end{equation}
\end{lemma}
\begin{proposition}
\label{prop:ZX}
For $n\geqslant 1$, we have
\begin{align}
  \label{eqn:ZX}
  Z_n & = \Bigg(\prod_{i=1}^n\vartheta_4(2(\eta+x_i),p^2)\vartheta_1(\eta+x_i)\vartheta_1(\eta-x_i)\Bigg)X_n,\\
    \label{eqn:ZBarX}
    \bar Z_n^\pm &= \Bigg(\prod_{i=1}^n\vartheta_1(2(\eta+x_i),p^2)\Bigg)\bar X_n^\pm,
\end{align}
where $X_n=X_n(x_1,\dots,x_n)$ and $\bar X_n^\pm=\bar X_n^\pm (x_1,\dots,x_n)$ are even theta functions of degree $2(n-1)$ and $2n$, respectively, nome $p$ and norm $0$ with respect to $x_i$ for each $i=1,\dots,n$. For $n\geqslant 2$, $X_n$ and $\bar X_n^\pm$ are symmetric functions in $x_1,\dots,x_n$.
\end{proposition}
\begin{proof}
  We present the proof of \eqref{eqn:ZX}. According to \cref{prop:PerZ}, $Z_n$ is a theta function with respect to $x_i$ with degree $2(n+1)$, nome $p$ and norm $t=\pi/2+\eta$. By \cref{prop:TrivialZeroes} it vanishes if $x_i= -\beta_1,+\beta_1,-\beta_4$ and $-\beta_3+\pi+\pi \tau$ (where we used the pseudo-periodicity of the scalar product). Hence, we may apply \cref{lem:Factorisation} to write $Z_n$ as a product of the trivial factors  
  \begin{multline}
    \vartheta_1(x_i+\beta_1)\vartheta_1(x_i-\beta_1)\vartheta_1(x_i+\beta_4)\vartheta_1(x_i + \beta_3-\pi-\pi \tau)\\
      = B \vartheta_4(2(x_i+\eta),p^2)\vartheta_1(\eta+x_i)\vartheta_1(\eta-x_i),
   \end{multline}
  where $B = -\i p^{-1/2}\vartheta_4(0,p^2)$, and a theta function with respect to $x_i$ of degree $2(n-1)$, nome $p$ and norm $t=0$. Since this holds for each $i=1,\dots,n$, we obtain the factorisation \eqref{eqn:ZX}, absorbing a power of the constant $B$ into the definition of $X_n$. The symmetry of $X_n$ for $n\geqslant 2$ follows from \cref{prop:SymmetryZ} and the symmetry of the factorised expression. Moreover, it follows from \cref{prop:InvZ} that $X_n$ is an even function of each $x_i$.
  
  The proof of \eqref{eqn:ZBarX} is similar.
\end{proof}

Henceforth, we analyse the properties of the functions $X_n,\,\bar X^\pm_n$. For coherence, we define 
\begin{equation}
\label{eqn:X0}
  X_0 = Z_0 = 1, \quad \bar X^+_0=\bar Z_0^+ = 2, \quad \bar X^-_0=\bar Z_0^- = 0.
\end{equation}
We note that for $n=1$, the expressions \eqref{eqn:Zn1} lead to
\begin{align}
  \begin{split}
  \label{eqn:X1}
  X_1 &= \frac{1}{2}\rho \vartheta_2(0)\vartheta_{2}(\eta+\lambda),\\
  \bar X_1^+ &= \rho \vartheta_4(0)\vartheta_4(\eta+\lambda)\vartheta_3(\eta+x_1)\vartheta_3(\eta-x_1),\\
   \bar X_1^- &= \rho \vartheta_3(0)\vartheta_3(\eta+\lambda)\vartheta_4(\eta+x_1)\vartheta_4(\eta-x_1),
   \end{split}
\end{align}
where $\rho$ is defined in \eqref{eqn:NormFactor}.

\subsubsection{Reduction relations}

The reduction relations for the vector $|\Psi_n\rangle$, given in \cref{lem:Reduction}, lead to reduction relations for the functions $X_n$ and $\bar X_n^\pm$. To state them, we define 
\begin{align}
    F(x)&=\frac{\vartheta_2(x)\vartheta_2(x+\eta)\vartheta_1(x+\lambda)\vartheta_{1}(x-\lambda+\eta)}{\vartheta_4(0,p^2)^2},\label{eqn:F}\\
    \bar F(x)&=\frac{\vartheta_3(x)\vartheta_3(x+\eta)\vartheta_4(x)\vartheta_4(x+\eta)\vartheta_1(x-\eta)^2\vartheta_1(x+\lambda)\vartheta_{1}(x-\lambda+\eta)}{\vartheta_4(0,p^2)^2}.\label{eqn:Fbar}
  \end{align}
\begin{proposition}
\label{prop:Red}
  For $n\geqslant 2$ and for each $2\leqslant i \leqslant n$, we have the reduction relations
  \begin{multline}
    X_n(x_1 =x_i+\eta,\dots,x_i,\dots,x_n) = \,F(x_i) \prod_{j=2,j\neq i}^n\vartheta_1(x_i-x_j-\eta)^2\vartheta_1(x_i+x_j-\eta)^2\\
   \times X_{n-2}(x_2,\dots,x_{i-1},x_{i+1},\dots,x_{n}),
  \end{multline}
  and
  \begin{multline}
    \bar X_n^\pm (x_1 =x_i+\eta,\dots,x_i,\dots,x_n) = \,\bar F(x_i) \prod_{j=2,j\neq i}^n\vartheta_1(x_i-x_j-\eta)^2\vartheta_1(x_i+x_j-\eta)^2\\
     \times \bar X_{n-2}^\pm(x_2,\dots,x_{i-1},x_{i+1},\dots,x_{n}).
  \end{multline}
  \begin{proof}
    We focus on the proof of the reduction relations for $X_n$. To this end, we abbreviate $Z_n' = Z_n(x,x+2\eta,x_3,\dots,x_n)$. Using the definition of $Z_n$ and \cref{lem:Exchange}, we may write
    \begin{equation}
      Z_n' = \frac{\langle \xi_{n}(x,x+2\eta,\dots)|\check R_{2,3}(-2(x+\eta))|\Psi_n(x,x+2\eta,-x,-x-2\eta,\dots)\rangle}{r(-2(x+\eta))}.
    \end{equation}
    The first two arguments $x,x+2\eta$ of the vector $|\Psi_n\rangle$ allow us to apply
 the reduction relation of \cref{lem:Reduction}. We obtain
    \begin{multline}
      Z_n' = \vartheta_1(2x)\vartheta_1(2(x-\eta))\vartheta_1(x-2\eta)\prod_{i=3}^n\vartheta_1(x-x_i-2\eta)\vartheta_1(x+x_i-2\eta)\\
      \times \frac{\langle \xi_{n}(x,x+2\eta,\dots)|\check R_{2,3}(-2(x+\eta))\left(|s\rangle \otimes |\Psi_{n-1}(-x,-x-2\eta,\dots)\rangle\right)}{r(-2(x+\eta))}.
    \end{multline}
    Next, we use $\check R(-2\eta)|s\rangle = -2r(-2\eta)|s\rangle$ to write
    \begin{align}
    \begin{split}
      Z_n' =& \,\vartheta_1(2x)\vartheta_1(2(x-\eta))\vartheta_1(x-2\eta)\prod_{i=3}^n\vartheta_1(x-x_i-2\eta)\vartheta_1(x+x_i-2\eta)\\
      &\times \frac{\langle \xi_{n}(x,x+2\eta,\dots)|\check R_{2,3}(-2(x+\eta))\check R_{1,2}(-2\eta)\left(|s\rangle \otimes |\Psi_{n-1}(-x,-x-2\eta,\dots)\rangle\right)}{-2r(-2\eta)r(-2(x+\eta))}.
    \end{split}
\end{align}
With the help of the boundary Yang-Baxter equation \eqref{eqn:bYBE} and the symmetry of the $\check R$-matrix, we may write
\begin{align}
     \begin{split}
      Z_n'=& \,\vartheta_1(2x)\vartheta_1(2(x-\eta))\vartheta_1(x-2\eta)\prod_{i=3}^n\vartheta_1(x-x_i-2\eta)\vartheta_1(x+x_i-2\eta)\\
      &\times \frac{\langle \xi_{n}(x+2\eta,x,\dots)|\check R_{2,3}(-2(x+\eta))\check R_{3,4}(-2\eta)\left(|s\rangle \otimes |\Psi_{n-1}(-x,-x-2\eta,\dots)\rangle\right)}{-2r(-2\eta)r(-2(x+\eta))}
    \end{split}\\
     \begin{split}
      =& \,\vartheta_1(2x)\vartheta_1(2(x-\eta))\vartheta_1(x-2\eta)\prod_{i=3}^n\vartheta_1(x-x_i-2\eta)\vartheta_1(x+x_i-2\eta)\\
      &\times \frac{\langle \xi_{n}(x+2\eta,x,\dots)|\check R_{2,3}(-2(x+\eta))\left(|s\rangle \otimes |\Psi_{n-1}(-x-2\eta,-x,\dots)\rangle\right)}{-2r(-2(x+\eta))}.
    \end{split}
    \end{align}
  The last equality follows from an application of \cref{lem:Exchange}.
   The resulting expression suggests yet another application of \cref{lem:Reduction}. We obtain
    \begin{multline}
      Z_n'= \vartheta_1(2x)\vartheta_1(2(x-\eta))\vartheta_1(x-2\eta)^2\prod_{i=3}^n\vartheta_1(x-x_i-2\eta)^2\vartheta_1(x+x_i-2\eta)^2\\
      \times \frac{\left(\langle \chi(x+2\eta)|\otimes \langle \chi(x)|\right)\check R_{2,3}(-2(x+\eta))\left(|s\rangle \otimes |s\rangle \right)}{2r(-2(x+\eta))} Z_{n-2}(x_3,\dots,x_n).
    \end{multline}
    The scalar product in the second line of this equality is given in \cref{lem:SpecialSP}. From this scalar product and from the relation between $Z_n$ and $X_n$, given in \cref{prop:ZX}, we infer
    \begin{multline}
       X_n(x,x+2\eta,x_3,\dots,x_n) =  F(x-\eta)\prod_{i=3}^n\vartheta_1(x-x_i-2\eta)^2\vartheta_1(x+x_i-2\eta)^2\\
      \times X_{n-2}(x_3,\dots,x_n).
    \end{multline}
    Finally, we set $x = x_2+\eta$ and obtain the reduction relation for $X_n$ with $i=2$. For $i=3,\dots,n$, it follows from the symmetry of $X_n$. 
   
    The proof of the reduction relations for $\bar X_n^\pm$ is similar.
  \end{proof}
\end{proposition}

\begin{proposition}
\label{prop:Red2}
For each $n\geqslant 1$, we have the reduction relations
\begin{align}
\label{eqn:RedX2}
X_n(\beta_2,x_2,\dots,x_n) 
	&= \frac{(-1)^{n-1}\vartheta_2(\eta+\lambda)}{\vartheta_4(0,p^2)}\Bigg(\prod_{i=2}^{n-1}\vartheta_2(x_i)^2\Bigg)X_{n-1}(x_2,\dots,x_n),\\
\bar X_n^\pm(\beta_3,x_2,\dots,x_n) 
	&=-p^{-n/2}\ee^{-2\i n\eta}\frac{\vartheta_2(\eta)\vartheta_3(0)\vartheta_3(\eta+\lambda)}{\vartheta_4(0,p^2)}\Bigg(\prod_{i=2}^n\vartheta_3(x_i)^2\Bigg)\bar X^\mp_{n-1}(x_2,\dots,x_n),\\
\bar X_n^\pm(\beta_4,x_2,\dots,x_n) 
	&=-p^{-n/2}\ee^{-2\i n\eta}\frac{\vartheta_2(\eta)\vartheta_4(0)\vartheta_4(\eta+\lambda)}{\vartheta_4(0,p^2)}\Bigg(\prod_{i=2}^n\vartheta_4(x_i)^2\Bigg) \bar X^\pm_{n-1}(x_2,\dots,x_n).
\end{align}
\begin{proof}
  We focus on the reduction relation for $X_n$. To this end, we consider the vector $|\Psi_n\rangle$ with arguments as set in \eqref{eqn:DefZ}. We choose $x_1=-\beta_2$ and apply \cref{lem:Reduction}. This leads to
\begin{multline}
    |\Psi_n(-\beta_2,\beta_2,x_2,-x_2,\dots,x_n,-x_n)\rangle\\
     = -\vartheta_2(0)\prod_{i=2}^n\vartheta_2(x_i)^2\sigma_1^z|s\rangle\otimes P|\Psi_{n-1}(x_2,-x_2,\dots,x_n,-x_n)\rangle,
  \end{multline}
where $P$ is the spin-parity operator, defined in \eqref{eqn:DefP}.
We use this relation and $ \sigma_1^z \sigma_2^z|\chi(x)\rangle = -|\chi(x)\rangle$ to obtain
\begin{equation}
    Z_n(-\beta_2,x_2,\dots,x_n) = (-1)^{n-1}\vartheta_2(0)\prod_{i=2}^n\vartheta_2(x_i)^2\left(\langle \chi(-\beta_2)|\sigma_1^z|s\rangle\right)Z_{n-1}(x_2,\dots,x_n).
\end{equation}
Next, we evaluate the matrix element $\langle\chi(-\beta_2)|\sigma_1^z|s\rangle$ with the help of
\begin{align}
\langle \chi(x)|\sigma_1^z|s\rangle=\langle \chi(x)|\left(|{\uparrow\downarrow}\rangle+|{\downarrow\uparrow}\rangle \right)= \vartheta_2(\eta+\lambda)\vartheta_1(x-\eta).
\end{align}
We find
 \begin{align}
    Z_n(-\beta_2,x_2,\dots,x_n) = (-1)^{n-1}\vartheta_2(0)\vartheta_2(\eta)\vartheta_2(\eta+\lambda)\prod_{i=2}^n\vartheta_2(x_i)^2Z_{n-1}(x_2,\dots,x_n).
\end{align} 
We obtain \eqref{eqn:RedX2} from this equality and the relation between $Z_n$ and $X_n$, given in \cref{prop:ZX}. 
  
  To prove the two reduction relations for $\bar X_n^\pm$, one needs to use \cref{lem:SRPsi} and the scalar products
  \begin{align}
  \langle\bar \chi(x)|(|{\downarrow\downarrow}\rangle-|{\uparrow\uparrow}\rangle) &= \vartheta_4(\eta+\lambda)\vartheta_3(x-\eta),\\
       \langle\bar \chi(x)|(|{\downarrow\downarrow}\rangle+|{\uparrow\uparrow}\rangle) &= \vartheta_3(\eta+\lambda)\vartheta_4(x-\eta).
  \end{align}
 The calculation is similar to the proof for $X_n$.
\end{proof}
\end{proposition}

\subsection{Determinants}
\label{sec:Determinants}

In this section, we present our main results for the inhomogeneous supersymmetric eight-vertex model. They provide explicit expressions for the functions $X_{n}$ and $\bar X^\pm_n$, and hence for the scalar products $Z_n$ and $\bar Z_n^\pm$. These expressions are given in terms of the so-called elliptic Tsuchiya determinant, which originally appeared as a partition function for the SOS model \cite{filali:11}.

\subsubsection{The elliptic Tsuchiya determinant}

Let us introduce
\begin{equation}
  \mathbb h(x,y) = \vartheta_1(x-y+\eta)\vartheta_1(x-y-\eta)\vartheta_1(x+y+\eta)\vartheta_1(x+y-\eta).
\end{equation}
We define $\mathbb H_0 = 1$ and, for each $k\geqslant 1$, the function
\begin{equation}
  \label{eqn:DefTsuchiya}
  \mathbb H_{2k}(x_1,\dots x_k; x_{k+1},\dots,x_{2k})= \frac{\prod_{i,j=1}^k\mathbb h(x_i,x_{j+k})}{\mathbb \Delta(x_1,\dots,x_k)\mathbb\Delta(x_{k+1},\dots,x_{2k})}\det_{i,j=1}^k \left(\frac{1}{\mathbb h(x_i,x_{j+k})}\right),
\end{equation}
where
\begin{equation}
  \mathbb\Delta(x_1,\dots,x_k) = \prod_{1\leqslant i<j\leqslant k}\vartheta_1(x_j-x_i)\vartheta_1(x_j+x_i).
\end{equation}
For $k\geqslant 1$, $\mathbb H_{2k}(x_1,\dots x_k; x_{k+1},\dots,x_{2k})$ is an even theta function of degree $2(k-1)$ and norm $0$ with respect to each $x_i,\,i=1,\dots,2k$. Moreover, it clearly is separately symmetric in $x_1,\dots,x_k$ and $x_{k+1},\dots,x_{2k}$. Less obvious is that it is symmetric in all its variables \cite{zinnjustin:13}. Here, we present a simple proof based on determinant condensation \cite{bressoudbook}.
\begin{proposition}
\label{prop:TsuchiyaSymmetry}
For each $k\geqslant 1$, $\mathbb H_{2k}(x_1,\dots x_k; x_{k+1},\dots,x_{2k})$ is a symmetric function of $x_1,\dots,x_{2k}$.
\end{proposition}
\begin{proof}
  The case $k=1$ is trivial: We have $\mathbb H_2(x_1;x_2) = 1$, which is obviously symmetric in $x_1$ and $x_2$. Therefore, we consider $k\geqslant 2$. Let $A = (a_{ij})_{i,j=1}^k$ be a matrix with $a_{kk}\neq 0$. We have the condensation identity
  \begin{equation}
  \label{eqn:DetATrick}
    \det A = a_{kk}^{-(k-2)} \det_{i,j=1}^{k-1}
    \left|
    \begin{array}{cc}
      a_{ij} & a_{ik}\\
      a_{kj} & a_{kk}
    \end{array}
    \right|.
  \end{equation}
  We apply this identity to the determinant in \eqref{eqn:DefTsuchiya}, and obtain
\begin{multline}
  \label{eqn:DefTsuchiya2}
  \mathbb H_{2k}(x_1,\dots, x_k;x_{k+1},\dots,x_{2k})= \frac{\prod_{i,j=1}^{k-1}\mathbb h(x_i,x_{j+k})}{\mathbb\Delta(x_1,\dots,x_{k-1})\mathbb\Delta(x_{k+1},\dots,x_{2k-1})}\\
  \times \det_{i,j=1}^{k-1}\left(\frac{\mathbb H_4(x_i,x_k;x_j,x_{2k})}{\mathbb h(x_i,x_{j+k})}\right).
  \end{multline}
We note that
\begin{align}
  \begin{split}
  \mathbb H_4(x_i,x_k;x_j,x_{2k})-&\mathbb H_4(x_i,x_{2k};x_j,x_{k})\\
  &=-\frac{\vartheta_1(3\eta)\vartheta_1(x_i-x_j)\vartheta_1(x_i+x_j)\vartheta_1(x_k-x_{2k})\vartheta_1(x_k+x_{2k})}{\vartheta_1(\eta)}.
  \end{split}
    \end{align}
    The right-hand side vanishes since $\eta=\pi/3$. Therefore, we have
    \begin{equation}
      \mathbb H_4(x_i,x_k;x_j,x_{2k})=\mathbb H_4(x_i,x_{2k};x_j,x_{k}).
    \end{equation}
    Using \eqref{eqn:DefTsuchiya2}, we conclude that $\mathbb H_{2k}(x_1,\dots,x_k;x_{k+1},\dots, x_{2k})$ is symmetric in $x_k$ and $x_{2k}$. Since it is separately symmetric in $x_1,\dots,x_k$ and $x_{k+1},\dots,x_{2k}$, it is symmetric in all its variables.
\end{proof}

To stress the symmetry of the function $\mathbb H_{2k}(x_1,\dots,x_k;x_{k+1},\dots,x_{2k})$, we omit the semicolon and simply write $\mathbb H_{2k}(x_1,\dots,x_{2k})$. Finally, we note that in \cite{zinnjustin:13} the following reduction relation was established:
\begin{lemma}
\label{lem:HReduction}
For each $k\geqslant 1$ and each $i=2,\dots,2k$, we have
\begin{multline}
\mathbb H_{2k}(x_1 = x_i+\eta,\dots,x_i,\dots,x_{2k})=\prod_{j=2,j\neq i}^{2k}\vartheta_1(x_i-x_j-\eta)\vartheta_1(x_i+x_j-\eta)\\
\times \mathbb H_{2(k-1)}(x_2,\dots,x_{i-1},x_{i+1},\dots).
\end{multline}
\end{lemma}

\subsubsection{The functions $\bm{X_n}$ and $\bm{\bar X_n^\pm}$} 

Let $m\geqslant 1$ be an integer and $z_1,\dots,z_m,t$ be complex numbers. We say that $z_1,\dots,z_m$ are independent if \textit{(i)} $z_i-z_j \not\equiv 0 \,\,(\text{mod}\,\pi,\pi \tau)$ for all $1\leqslant i < j \leqslant m$ and \textit{(ii)} $z_1+\dots+z_m-t \not\equiv 0 \,\,(\text{mod}\,\pi,\pi \tau)$. We have the following property of theta functions \cite{weber:08}:
\begin{theorem}
	\label{thm:Id}
Let $f,g$ be theta functions of degree $m\geqslant 1$, nome $p$ and norm $t$. If there are $m$ independent complex numbers $z_1,\dots,z_m$ such that $f(z_i)=g(z_i)$ for each $i=1,\dots,m$, then $f=g$.
\end{theorem}
We use this theorem to find an explicit expression $X_n$ in terms of the elliptic Tsuchiya determinant. To this end, we define for each $k\geqslant 0$ the functions
\begin{align}
  Y_{2k} &= \frac{(-1)^k}{\vartheta_4(0,p^2)^{2k}}\mathbb H_{2k}(x_1,\dots,x_{2k})\mathbb H_{2(k+1)}(x_1\dots,x_{2k},\beta_2,\eta+\lambda),\\
  Y_{2k+1} &= \frac{(-1)^{k}\vartheta_2(\eta+\lambda)}{\vartheta_4(0,p^2)^{2k+1}}\mathbb H_{2(k+1)}(x_1,\dots,x_{2k+1},\beta_2)\mathbb H_{2(k+1)}(x_1\dots,x_{2k+1},\eta+\lambda).
\end{align}
\begin{lemma}
	\label{lem:RedY}
	The function $Y_n$ obeys the reduction relations of $X_n$ given in \cref{prop:Red,prop:Red2}.
	\end{lemma}
\begin{proof}
	The proof is a straightforward calculation using \cref{lem:HReduction}.
\end{proof}

\begin{theorem}
  \label{thm:X}
  For each $n\geqslant 0$, we have $X_n=Y_n$.
\begin{proof}
	The proof is based on a standard induction argument that goes back to Izergin and Korepin \cite{izergin:92,korepin:93}. We start by examining the base cases $n=0$ and $n=1$. From the expressions given above, we find
	\begin{equation}
    Y_0 = 1, \quad Y_1 = \frac{\vartheta_2(\eta+\lambda)}{\vartheta_4(0,p^2)}.
  \end{equation}
  They are equal to the expressions of $X_0$ and $X_1$, given in \eqref{eqn:X0} and \eqref{eqn:X1}, respectively.
    
   We now make the induction hypothesis that $X_n = Y_n$ for $n=m$ and $n=m-1$, where $m\geqslant 2$ is some integer. For the induction step, we consider $X_{m+1}$ and $Y_{m+1}$ as functions of $x_1$. By \cref{prop:ZX}, $X_{m+1}$ is a theta function of degree $2m$, nome $p$ and norm $0$. The same holds for $Y_{m+1}$, by the properties of the elliptic Tsuchiya determinant. By \cref{lem:RedY} both functions obey the same reduction relations. For $x_1=\pm (x_j\pm \eta),\, j=2,\dots,m+1$, they allow us to express $X_{m+1}$ and $Y_{m+1}$ in terms of $X_{m-1}$ and $Y_{m-1}$, respectively. Likewise, for $x_1=\pm \beta_2$, they lead to expressions of $X_{m+1}$ and $Y_{m+1}$ in terms of $X_m$ and $Y_m$ respectively. Using these expressions and the induction hypothesis, we conclude that $X_{m+1}=Y_{m+1}$ for $4m+2$ values of $x_1$. We choose $x_2,\dots,x_{m+1}$ so that they contain $2m$ independent values. By \cref{thm:Id}, we conclude that $X_{m+1}=Y_{m+1}$, and hence $X_n = Y_n$ for $n=m+1$ and $n=m$, which ends the induction step. The theorem follows.
\end{proof}
\end{theorem}

Next, we provide explicit expressions for the functions $\bar X_n^\pm$ in terms of the elliptic Tsuchyia determinant. For each $k\geqslant 0$, we define
\begin{multline}
	\bar Y_{2k}^\pm =  \gamma_{2k}^\pm \mathbb H_{2(k+1)}(x_1,\dots,x_{2k},0,\beta_4)\mathbb H_{2(k+1)}(x_1,\dots,x_{2k},\eta+\lambda,\beta_3)\\
	 + \delta_{2k}^\pm \mathbb H_{2(k+1)}(x_1,\dots,x_{2k},0,\beta_3) \mathbb H_{2(k+1)}(x_1,\dots,x_{2k},\eta+\lambda,\beta_4),
\end{multline}
and
\begin{multline}
	\bar Y_{2k+1}^\pm=\gamma_{2k+1}^\pm \mathbb H_{2(k+1)}(x_1,\dots,x_{2k+1},0)\mathbb H_{2(k+2)}(x_1,\dots,x_{2k+1},\eta+\lambda,\beta_3,\beta_4)\\
	+\delta_{2k+1}^\pm \mathbb H_{2(k+1)}(x_1,\dots,x_{2k+1},\eta+\lambda)\mathbb  H_{2(k+2)}(x_1,\dots,x_{2k+1},0,\beta_3,\beta_4).
\end{multline}
	Here, the coefficients $\gamma_{2k}^\pm,\, \gamma_{2k+1}^\pm$ and $\delta_{2k}^\pm,\, \delta_{2k+1}^\pm$ are given by
\begin{align}
	\gamma_{2k}^\pm &= \left(\frac{p\ee^{-2\i \eta }}{\vartheta_4(0,p^2)^{2}}\right)^k \gamma_0^\pm, &\delta_{2k}^\pm &= \left(\frac{p\ee^{-2\i \eta }}{\vartheta_4(0,p^2)^{2}}\right)^k \delta_0^\pm, \\
	\gamma_{2k+1}^\pm &= \left(\frac{p\ee^{-2\i \eta }}{\vartheta_4(0,p^2)^{2}}\right)^k \gamma_1^\pm, & \delta_{2k+1}^\pm &= \left(\frac{p\ee^{-2\i \eta }}{\vartheta_4(0,p^2)^{2}}\right)^k\delta_1^\pm,
\end{align}
	where
\begin{align}
	\gamma_0^+ &= 2\left(\frac{\vartheta_3(0)\vartheta_4(\eta+\lambda)}{\vartheta_2(0)\vartheta_1(\eta+\lambda)}\right)^2, & \delta_0^+ &= -2\left(\frac{\vartheta_4(0)\vartheta_3(\eta+\lambda)}{\vartheta_2(0)\vartheta_1(\eta+\lambda)}\right)^2,\\
	\gamma_0^- &= -\frac{2\vartheta_3(0)\vartheta_4(0)\vartheta_3(\eta+\lambda)\vartheta_4(\eta+\lambda)}{\vartheta_2(0)^2\vartheta_1(\eta+\lambda)^2},
	&\delta_0^- &= \frac{2\vartheta_3(0)\vartheta_4(0)\vartheta_3(\eta+\lambda)\vartheta_4(\eta+\lambda)}{\vartheta_2(0)^2\vartheta_1(\eta+\lambda)^2},
\end{align}
	and  
\begin{equation}
	\gamma_1^\pm = -\frac{p\ee^{-2\i\eta}\vartheta_4(0)}{\vartheta_4(0,p^2)\vartheta_4(\eta+\lambda)\vartheta_2(\eta)}\gamma_0^\pm,\quad \delta_1^\pm = -\frac{p\ee^{-2\i\eta}\vartheta_4(\eta+\lambda)}{\vartheta_4(0,p^2)\vartheta_4(0)\vartheta_2(\eta)}\delta_0^\pm.
\end{equation}

\begin{lemma}
	\label{lem:RedYBar}
	The functions $\bar Y_n^\pm$ satisfy the reduction relations of $\bar X_n^\pm$ given in \cref{prop:Red,prop:Red2}.
	\end{lemma}
\begin{proof}
	The proof is a straightforward calculation using \cref{lem:HReduction}.
\end{proof}

\begin{theorem}
\label{thm:XBar}
For each $n\geqslant 0$, we have $\bar Y_n^\pm = \bar X_n^\pm$.
\begin{proof}
  As for \cref{thm:X}, the proof is based on induction. Let us check the base cases
  $n=0$ and $n=1$, using the explicit expressions given above. For $n=0$, we find
  \begin{equation}
    \bar Y_0^+ = 2,\quad \bar Y_0^-=0.
  \end{equation}
  Furthermore, for $n=1$, we obtain
  \begin{equation}
    \bar Y_1^\pm = \gamma_1^\pm \mathbb H_4(x_1,\eta+\lambda,\beta_3,\beta_4)+\delta_1^\pm \mathbb H_4(x_1,0,\beta_3,\beta_4).
  \end{equation}
  We have \cite{zinnjustin:13}
  \begin{equation}
  \mathbb H_4(x,y,\beta_3,\beta_4) = -\frac{p^{-1}\ee^{2\i \eta}\vartheta_2(\eta)}{\vartheta_2(0)}\left(\vartheta_3(x+\eta)\vartheta_3(x-\eta)\vartheta_4(y)^2+\vartheta_4(x+\eta)\vartheta_4(x-\eta)\vartheta_3(y)^2\right).
\end{equation}
We use this expression, together with the definition of $\gamma_1^\pm,\,\delta_1^\pm$, and find
\begin{align}
 	\bar Y_1^+ &= \frac{2\vartheta_4(0)\vartheta_4(\eta+\lambda)\vartheta_3(\eta+x_1)\vartheta_3(\eta-x_1)}{\vartheta_2(0)\vartheta_4(0,p^2)},\\ \quad 
 	\bar Y_1^- &= \frac{2\vartheta_3(0)\vartheta_3(\eta+\lambda)\vartheta_4(\eta+x_1)\vartheta_4(\eta-x_1)}{\vartheta_2(0)\vartheta_4(0,p^2)}.
\end{align}
We compare our findings with \eqref{eqn:X0} and \eqref{eqn:X1} and conclude that $\bar Y_n^\pm=\bar X_n^\pm$ for $n=0,1$.

Next, we make the induction hypothesis that $\bar X_n^\pm = \bar Y_n^\pm$ for $n=m$ and $n=m-1$, where $m\geqslant 2$ is some integer. For the induction step, we essentially follow the proof of \cref{thm:X}. The only differences are \textit{(i)} that we need to use \cref{lem:RedYBar} and \textit{(ii)} that the degree of the theta functions is now $2(m+1)$. Nonetheless, we may choose $x_2,\dots,x_{m+1}$ to find enough independent points for the application of \cref{thm:Id} to be possible. It allows us to conclude that $\bar X_{m+1}^\pm = \bar Y_{m+1}^\pm$. The induction step follows, which ends the proof.
\end{proof}
\end{theorem}

\subsection{Polynomials}
\label{sec:RGPolynomials}

In this section, we recall the relation between the elliptic Tsuchiya determinant and a family of polynomials introduced by Zinn-Justin and Rosengren.
Moreover, we discuss the properties of these polynomials that are relevant to our forthcoming analysis of the scalar products. These properties can either be found in or easily derived from \cite{zinnjustin:13,rosengren:13,rosengren:13_2,rosengren:14,rosengren:15}, but we prefer to include them here for completeness.

\subsubsection{Definition and relation to the elliptic Tsuchiya determinant}

Let us introduce
\begin{equation}
h(w,w')=1-(3+\zeta^2)ww'+(1-\zeta^2)ww'(w+w').
\end{equation}
We define $H_0 = 1$ and, for each $k\geqslant 1$, the function\footnote{We use Zinn-Justin's notations. The polynomials $H_{2k}(w_1,\dots,w_{2k})$ are related to Rosengren's polynomials $T(x_1,\dots,x_{2k})$ by a change of variables \cite{rosengren:13}.}
\begin{equation}
\label{eqn:RGPolynomials}
	H_{2k}(w_1,\dots,w_{2k}) = \frac{\prod_{i,j=1}^k h(w_i,w_{j+k})}{\Delta(w_1,\dots,w_k)\Delta(w_{k+1},\dots,w_{2k})} \det_{i,j=1}^k\left(\frac{1}{h(w_i,w_{j+k})}\right),
\end{equation}
where $\Delta(w_1,\dots,w_m)= \prod_{1\leqslant i < j \leqslant m} (w_j-w_i)$ denotes the Vandermonde determinant in $m$ variables. 
Clearly, $H_{2k}(w_1,\dots,w_{2k})$ is a polynomial in $w_1,\dots,w_{2k}$ and $\zeta$.

In the next proposition, we recall its relation to the elliptic Tsuchiya determinant. To this end, we use the elliptic function
\begin{equation}
  \label{eqn:DefW}
  w(x) = \frac{\vartheta_4(\eta,p^2)}{\vartheta_4(0,p^2)}\frac{\vartheta_1(x)^2}{\vartheta_1(x-\eta)\vartheta_1(x+\eta)}.
\end{equation}
Moreover, in the following we often use a parameterisation of $\zeta$ in terms of the elliptic nome $p$:
\begin{equation}
  \label{eqn:ZetaP}
  \zeta = \left(\frac{\vartheta_1(\eta,p^2)}{\vartheta_4(\eta,p^2)}\right)^2.
\end{equation}
\begin{lemma}
\label{lem:UniformisationH}
  Let $k\geqslant 1$ and  $w_i = w(x_i)$ for each $i=1,\dots,2k$. If \eqref{eqn:ZetaP} holds, then we have
  \begin{align}
    \label{eqn:Unif1}
    \mathbb H_{2k}(x_1,\dots,x_{2k}) = f(x_1,\dots,x_{2k})H_{2k}(w_1,\dots,w_{2k}),
  \end{align}
  where 
  \begin{equation}
  \label{eqn:Unif2}
    f(x_1,\dots,x_{2k}) = \left(\frac{\vartheta_4(\eta,p^2)}{\vartheta_1(\eta)^2\vartheta_4(0,p^2)}\right)^{k(k-1)}\prod_{i=1}^{2k}\left(\vartheta_1(x_i+\eta)\vartheta_1(x_i-\eta)\right)^{k-1}.
  \end{equation}
  \begin{proof}
  The key observations are the two relations
  \begin{equation}
   w(x)-w(y)=-\frac{\vartheta_4(\eta,p^2)\vartheta _1(\eta)^2\vartheta_1(x-y)\vartheta_1(x+y)}{\vartheta_4(0,p^2)\vartheta_1 (x-\eta)\vartheta_1(x+\eta)\vartheta_1(y-\eta)\vartheta_1(y+\eta)},
  \end{equation}
  and
  \begin{equation}
    h(w(x),w(y)) = \frac{\vartheta_1(\eta)^4\mathbb h(x,y)
}{\left(\vartheta_1 (x-\eta)\vartheta_1(x+\eta)\vartheta_1(y-\eta)\vartheta_1(y+\eta)\right)^2}.
  \end{equation}
  They follow from standard identities between Jacobi theta functions.
  Using them together with the definition of $H_{2k}(w_1,\dots,w_{2k})$ and $\mathbb H_{2k}(x_1,\dots,x_{2k})$ leads to \eqref{eqn:Unif1} and \eqref{eqn:Unif2}.
  \end{proof}
\end{lemma}
The symmetry of the elliptic Tsuchiya determinant implies that $H_{2k}(w_1,\dots,w_{2k})$ is a symmetric polynomial in $w_1,\dots,w_{2k}$ for $k\geqslant 1$. In the following, we use for each $k\geqslant 1$ and $j=1,\dots,2k-1$ the abbreviation
\begin{equation}
  H_{2k}(w_1,\dots,w_j) 
\equiv H_{2k}(w_1,\dots,w_j,0,\dots,0).
\end{equation}
Clearly, if $j\geqslant 2$ then $H_{2k}(w_1,\dots,w_j)$ is a symmetric polynomial in $w_1,\dots,w_j$.

\subsubsection{Determinant formulas}

We combine the determinant formula \eqref{eqn:DefTsuchiya2} for the elliptic Tsuchiya determinant with \cref{lem:UniformisationH}. This leads to 
\begin{multline}
 \label{eqn:RGAlt}
  H_{2k}(w_1,\dots,w_{2k})=\frac{\prod_{i,j=1}^{k-1} h(w_i,w_{j+k})}{\Delta(w_1,\dots,w_{k-1})\Delta(w_{k+1},\dots,w_{2k-1})}\\
  \times \det_{i,j=1}^{k-1}\left(\frac{H_4(w_i,w_k,w_{j+k},w_{2k})}{h(w_i,w_{j+k})}\right),
\end{multline}
for each $k\geqslant 2$, where
\begin{equation}
  \label{eqn:H4}
  H_4(w_1,\dots,w_4) = 3+\zeta^2+(\zeta^2-1)(w_1+w_2+w_3+w_4+(\zeta^2-1)w_1w_2w_3w_4).
\end{equation}
The determinant formula \eqref{eqn:RGAlt} allows us to infer several properties of the symmetric polynomial $H_{2k}(w_1,\dots,w_{2k})$ for $k>2$ from $k=2$. For instance, we easily find the degree of this polynomial with respect to each of its variables.
\begin{lemma}
\label{lem:HInfinite}
  For each $k\geqslant 1$, $H_{2k}(w_1,\dots,w_{2k})$ is a polynomial of degree $k-1$ in $w_i$ for each $i=1,\dots,2k$. Moreover, we have
  \begin{equation}
    \lim_{w_{2k-1}\to \infty} w^{-(k-1)}_{2k-1}H_{2k}(w_1,\dots,w_{2k-1},0) = (\zeta^2-1)^{k-1}H_{2(k-1)}(w_1,\dots,w_{2(k-1)}).
  \end{equation}
\end{lemma}
Furthermore, \eqref{eqn:RGAlt} allows us to find a useful formula for $H_{2k}(w,w')$ in terms of a determinant of a matrix with polynomial entries in $w,w'$ and $\zeta$. 
Indeed, from the method of divided differences \cite{behrend:12} we obtain:

\begin{lemma}
\label{lem:HDetFormula2}
For each $k\geqslant 2$, we have
\begin{equation}
   H_{2k}(w,w') = \det_{i,j=0}^{k-2}\left(H_4(w,w')\eta_{i,j}+(\zeta^2-1)\left(\eta_{i-1,j}+\eta_{i,j-1}+(\zeta^2-1)w w'\eta_{i-1,j-1}\right)\right),
\end{equation}
where $\eta_{i,j}$ is a polynomial in $\zeta$ with integer coefficients, given by
\begin{equation}
  \eta_{i,j} = \sum_{n=\left\lceil\frac{i+j}{3}\right\rceil}^{\min(i,j)}\frac{n!(3+\zeta^2)^{3n-(i+j)}(\zeta^2-1)^{i+j-2n}}{(i-n)!(j-n)!(3n-(i+j))!}, 
\end{equation}
if $i,j\geqslant 0$, and $\eta_{i,j}=0$ if $i<0$ or $j<0$.
\end{lemma}
The formula given in this lemma is quite practical to explicitly compute the polynomials $H_{2k}(w,w')$ with \textsc{Mathematica}.

\subsubsection{Bilinear identities}

The polynomials $H_{2k}(w_1,\dots,w_{2k})$ satisfy bilinear identities. They follow from a Pl\"ucker relation and the well-known Desnanot-Jacobi identity \cite{bressoudbook,noumi:04}, applied to the determinant in \eqref{eqn:RGPolynomials}.
\begin{lemma}
  \label{lem:Bilinear}
  For each $k\geqslant 1$ and all $x,y,u,v$, we have
  \begin{align}
  \label{eqn:Bilinear1}
  \begin{split}
  &(x-y)h(u,v)H_{2(k+1)}(w_1,\dots,w_{2k-1},x,y,v)H_{2k}(w_1,\dots,w_{2k-1},u)\\
 +&(y-u)h(x,v) H_{2(k+1)}(w_1,\dots,w_{2k-1},y,u,v)H_{2k}(w_1,\dots,w_{2k-1},x)\\
 +&(u-x)h(y,v) H_{2(k+1)}(w_1,\dots,w_{2k-1},u,x,v)H_{2k}(w_1,\dots,w_{2k-1},y)=0,
 \end{split} 
  \end{align}
   and 
   \begin{align}
  \label{eqn:Bilinear2}
    \begin{split}
    &(x-u)(y-v)H_{2(k+2)}(w_1,\dots,w_{2k},x,y,u,v)H_{2k}(w_1,\dots,w_{2k})\\
    = & \,h(x,v)h(y,u)H_{2(k+1)}(w_1,\dots,w_{2k},u,v)H_{2(k+1)}(w_1,\dots,w_{2k},x,y)\\
     &  -h(x,y)h(u,v)H_{2(k+1)}(w_1,\dots,w_{2k},y,u)H_{2(k+1)}(w_1,\dots,w_{2k},x,v).
  \end{split}
  \end{align}

\end{lemma}

\subsubsection{Fractional linear transformations}

Finally, the polynomial $H_{2k}(w_1,\dots,w_{2k})$ has a simple transformation behaviour under the fractional linear transformation $\zeta \to \zeta'$, where
\begin{equation}
  \zeta' = \frac{\zeta+3}{\zeta-1}.
\end{equation}
It follows from the property $h(w,w')|_{\zeta\to \zeta'} = h(2w/(\zeta-1),2w'/(\zeta-1))$:
\begin{lemma}
\label{lem:HTransform1}
  For each $k\geqslant 0$, we have
  \begin{align}
     \left.H_{2k}\left(w_1,\dots,w_{2k}\right)\right|_{\zeta\to \zeta'}
    &=\left(\frac{2}{\zeta-1}\right)^{k(k-1)}H_{2k}\left(\frac{2w_1}{\zeta-1},\dots,\frac{2w_{2k}}{\zeta-1}\right).
  \end{align}
\end{lemma}

\section{The homogeneous limit}
\label{sec:HomogeneousLimit}

In this section, we consider the case where the inhomogeneity parameters take the values $u_1=\dots=u_{2n+1}=0$. This is commonly referred to as the \textit{homogeneous limit}. To study the homogeneous limit of the vector $|\Psi_n\rangle$, it is convenient to define \begin{equation}
  \label{eqn:HomLimitPsi}
  |\psi_n\rangle = {\mathcal N}_n |\Psi_n(0,\dots,0)\rangle, \quad |\bar \psi_n\rangle = {\mathcal N}_n P|\Psi_n(0,\dots,0)\rangle.
\end{equation}
Here, $\mathcal N_n$ is a normalisation factor. We choose
\begin{equation}
  \label{eqn:DefN}
   \mathcal N_n = \frac{(-1)^{\lfloor \frac{n+1}{2}\rfloor}}{\vartheta_1(\eta)^{n^2}}\left(\frac{\vartheta_4(0,p^2)}{\vartheta_4(\eta,p^2)}\right)^{n(n+1)/2}.
\end{equation}
We shall see below that the vectors $|\psi_n\rangle$ and $|\bar\psi_n\rangle$ are non-vanishing. Their definition thus
 implies that they are linearly independent. They constitute a basis of the eigenspace of the eigenvalue $\Theta_n =(a+b)^{2n+1}$ of the transfer matrix of the homogeneous eight-vertex model \cite{hagendorf:18}.

The main goal of this section is to study the basis vectors $|\psi_n\rangle$ and $|\bar\psi_n\rangle$ with the help of the homogeneous limits of the scalar products $Z_n$ and $\bar Z_n^\pm$. We compute these homogeneous limits in \cref{sec:HomLimitZ} and \cref{sec:HomLimitZBar}, respectively. This computation allows us to obtain and analyse several components  of the basis vectors. In \cref{sec:SumOfComponents}, we exploit the relation between the eight-vertex model and the XYZ spin chain to compute the sum of the components of the basis vectors. In \cref{sec:Discussion}, we compare our results to several conjectures by Bazhanov and Mangazeev, and Razumov and Stroganov.

\subsection{The homogeneous limit of \texorpdfstring{$\boldsymbol{Z_n}$}{$Z_n$} }
\label{sec:HomLimitZ}

In this section, we compute the scalar product
\begin{equation}
   \label{eqn:DefS}
   S_n  = \left(\left(\langle{\uparrow\downarrow}|+\mu\langle{\downarrow\uparrow}|\right)^{\otimes n}\otimes\langle{\uparrow}|\right)|\psi_{n}\rangle
\end{equation}
from the homogeneous limit of $Z_n$. The results of our computations naturally depend on the elliptic nome $p$. We restrict our considerations to real $0 < p < 1$ and state our results in terms of the variable $\zeta$, defined in \eqref{eqn:ZetaP}, and
\begin{equation}
  \label{eqn:DefJ}
  J_2 = -\frac{1}{2}, \quad J_3 = \frac{1}{1+\zeta}, \quad J_4 = \frac{1}{1-\zeta}.
\end{equation}
We have $J_k = w(\beta_k)$ for $k=2,3,4$, where $w(x)$ is defined in \eqref{eqn:DefW}.
We note that the restriction of the range of the elliptic nome implies that $0 < \zeta <1$.

We divide the section into three parts. First, we find a closed-form expression for $S_n$ in terms of the polynomials defined in \cref{sec:RGPolynomials}. Second, we use this expression to compute and analyse the component $(\psi_n)_{\uparrow\downarrow\cdots\uparrow \downarrow \uparrow}$. Third, we evaluate the scalar product $S_n$ in the so-called trigonometric limit. 

\subsubsection{Closed-form expression}
\begin{theorem}
  \label{thm:S}
  For each $k\geqslant 0$, $S_{2k}$ and $S_{2k+1}$ are polynomials in $\mu$ and $\zeta$, given by
  \begin{align}
    \label{eqn:S2k}
    S_{2k} &= (2\mu)^k H_{2k}H_{2(k+1)}(J_2,\bar \mu),\\
    \label{eqn:S2k1}
    S_{2k+1} &= (2\mu)^k (\mu+1)H_{2(k+1)}(J_2)H_{2(k+1)}(\bar \mu),
  \end{align}
  where $\bar \mu = (\mu-1)^2/(\zeta^2-1)\mu$.
\end{theorem}

\begin{proof}
The proof has two parts. In part 1, we establish the explicit formulas for the scalar products. In part 2, we show that these expressions yield polynomials in both $\mu$ and $\zeta$.

\medskip

\noindent \textit{Part 1: Explicit formulas.} First, we evaluate the scalar product $Z_n$ for $x_1=\dots=x_n=0$. Using \eqref{eqn:DefZ}, this evaluation leads to
\begin{equation}
  \label{eqn:ZnHomLimit}
  Z_n(0,\dots,0) = (\vartheta_1(\lambda,p^2)\vartheta_4(\lambda-\eta,p^2))^n\left(\left(\langle{\uparrow\downarrow}|+\mu\langle{\downarrow\uparrow}|\right)^{\otimes n}\otimes\langle{\uparrow}|\right)|\Psi_{n}(0,\dots,0)\rangle,
\end{equation}
where
\begin{equation}
  \mu =\frac{\vartheta _4\left(\lambda ,p^2\right) \vartheta _1\left(\lambda -\eta ,p^2\right)}{\vartheta _1\left(\lambda ,p^2\right) \vartheta _4\left(\lambda -\eta ,p^2\right)}.
  \label{eqn:DefMu}
\end{equation}

We use \cref{prop:ZX} to write the left-hand side of \eqref{eqn:ZnHomLimit} in terms of $X_n$. Moreover, we use \eqref{eqn:HomLimitPsi} and \eqref{eqn:DefS} to rewrite its right-hand side in terms of $S_n$. Solving for the latter, we obtain
\begin{equation}
 S_n = \left(\frac{\vartheta_4(\eta,p^2)\vartheta_1(\eta)^{2}}{\vartheta_1(\lambda,p^2)\vartheta_4(\lambda-\eta,p^2)}\right)^n \mathcal N_nX_n(0,\dots,0).
\end{equation}

Second, we replace $\mathcal N_n$ by its definition \eqref{eqn:DefN}, and $X_{n}$ by its explicit expression in terms of the elliptic Tsuchiya determinant, given in \cref{thm:X}. This step requires to consider the cases $n=2k$ and $n=2k+1$ separately. We focus on the case $n=2k$, where we obtain
\begin{multline}
 S_{2k} = \frac{1}{\vartheta_1(\zeta)^{4k(k-1)}\left(\vartheta_1(\lambda,p^2)\vartheta_4(\lambda-\eta,p^2)\right)^{2k}} \left(\frac{\vartheta_4(0,p^2)}{\vartheta_4(\eta,p^2)}\right)^{(2k-1)k}\\
 \times\mathbb H_{2k}(0,\dots,0)\mathbb H_{2(k+1)}(0,\dots,\beta_2,\eta+\lambda).
\end{multline}
We now use \cref{lem:UniformisationH} to rewrite the Tsuchiya determinants in terms of the polynomials of \cref{sec:RGPolynomials}. After some algebra, we obtain
\begin{equation}
  \label{eqn:S2kIntermediate}
  S_{2k}=\left(\frac{\vartheta_1(\lambda-\eta)\vartheta_1(\lambda)\vartheta_2(\eta)\vartheta_2(0)\vartheta_4(\eta,p^2)}{\vartheta_1(\lambda,p^2)^2\vartheta_4(\lambda-\eta,p^2)^2\vartheta_4(0,p^2)}\right)^{k}H_{2k}H_{2(k+1)}(w(\beta_2),w(\lambda+\eta)).
\end{equation}
Here, $w=w(x)$ is the elliptic function defined in \eqref{eqn:DefW}. Moreover, the polynomials on the right-hand side of this equality implicitly depend on $\zeta$, which is given in term of $p$ by \eqref{eqn:ZetaP}.

Third, we simplify \eqref{eqn:S2kIntermediate} with the help of several classical identities between the Jacobi theta functions. For the prefactor, 
we use
\begin{equation}
  \vartheta_1(x)\vartheta_2(0) =2\vartheta_1(x,p^2)\vartheta_4(x,p^2), \quad \vartheta_2(\eta)\vartheta_4(\eta,p^2) = 2 \vartheta_2(0)\vartheta_4(0,p^2),
\end{equation}
and the relation \eqref{eqn:DefMu}. Moreover, the arguments of the polynomial are given by
\begin{equation}
  w(\beta_2) = J_2, \quad w(\eta+\lambda) = \bar \mu.
\end{equation}
Applying these identities, we obtain \eqref{eqn:S2k}.

The proof of \eqref{eqn:S2k1} is similar.

\medskip

\textit{Part 2: Polynomial nature.} \Cref{lem:HDetFormula2} implies that both $H_{2k}$ and $H_{2(k+1)}(J_2)$ are polynomials in $\zeta$. Moreover, one checks that $\mu H_4(w,\bar \mu)$ is a polynomial in $w,\mu$ and $\zeta$. By \cref{lem:HDetFormula2}, $\mu^k H_{2(k+1)}(w,\bar \mu)$ is a polynomial in $w,\mu$, and $\zeta$. Hence, both $\mu^k H_{2(k+1)}(\bar \mu)$ and $\mu^k H_{2(k+1)}(J_2,\bar \mu)$ are polynomials in $\mu$ and $\zeta$. The polynomial nature of $S_{2k}$ and $S_{2k+1}$ follows.

\medskip

Finally, we comment on the parameterisation \eqref{eqn:DefMu} used in this proof. By the properties of the Jacobi theta functions, $\mu$ is a meromorphic function of $\lambda$ for all $0<p<1$. 
The pole structure of this function implies that each real $\mu$ is the image of some $0 < \lambda < \pi$ under this mapping. 
Hence, the statements of the theorem hold for all real $\mu$. By analytic continuation, they hold for all complex $\mu$.

\end{proof}

\subsubsection{Components}
\begin{proposition}
\label{prop:AlternatingComponent}
For each $k\geqslant 0$, we have the components
\begin{align}
\label{eqn:PsiAlt2k}
(\psi_{2k})_{\uparrow\downarrow\cdots\uparrow\downarrow\uparrow } &= 2^k H_{2k}H_{2k}(J_2),
\\
\label{eqn:PsiAlt2k1}
(\psi_{2k+1})_{\uparrow\downarrow\cdots\uparrow\downarrow\uparrow } &= 2^{k}H_{2k}H_{2(k+1)}(J_2).
\end{align}
\end{proposition}
\begin{proof}
We present the proof of \eqref{eqn:PsiAlt2k}. For $k=0$ it is trivial. Hence, we consider $k\geqslant 1$. We have
\begin{equation}
 (\psi_{2k})_{\uparrow\downarrow\cdots\uparrow\downarrow\uparrow}= \left.S_{2k}\right|_{\mu=0}= 2^kH_{2k} \lim_{\mu \to 0}\mu^{k}H_{2(k+1)}(J_2,\bar \mu).
\end{equation}
We compute the limit on the right-hand side with the help of \cref{lem:HInfinite}, and find
\begin{equation}
  \lim_{\mu \to 0}\mu^{k}H_{2(k+1)}(J_2,\bar \mu) = (\zeta^2-1)^{-k}\lim_{\bar \mu \to \infty} \bar \mu^{-k}H_{2(k+1)}(J_2,\bar \mu)=H_{2k}(J_2).
\end{equation}
This result leads to \eqref{eqn:PsiAlt2k}. The derivation of \eqref{eqn:PsiAlt2k1} is similar.
\end{proof}

This proposition implies that the component $(\psi_{n})_{\uparrow\downarrow\cdots\uparrow\downarrow\uparrow }$ is a polynomial in the variable $\zeta$ with integer coefficients. We now compute the coefficients of its lowest-order and highest-order term.
To this end, we recall that the number of alternating sign matrices of size $n$ is given by \cite{bressoudbook}
\begin{equation}
  A(n)=\prod_{i=0}^{n-1}\frac{(3i+1)!}{(n+i)!}.
\end{equation}
In the following, we also frequently use the number $A_{\text{\tiny V}}(2n+1)$ of vertically-symmetric alternating sign matrices of size $2n+1$, and the number $N_8(2n)$ of cyclically-symmetric transpose complement plane partitions in a $2n\times 2n \times 2n$ cube \cite{bressoudbook,kuperberg:02}:
  \begin{equation}
  \label{eqn:DefAVN8}
  A_{\text{\rm \tiny V}}(2k+1)= \frac{1}{2^k}\prod_{i=1}^k \frac{(6i-2)!(2i-1)!}{(4i-1)!(4i-2)!}, \quad N_{8}(2n)=\prod_{i=0}^{n-1} \frac{(3i+1)(6i)!(2i)!}{(4i)!(4i+1)!}.
\end{equation}

\begin{proposition}
  \label{prop:AltComponentComb}
  For each $k\geqslant 0$, we have
\begin{align}
(\psi_{2k})_{\uparrow\downarrow\cdots\uparrow\downarrow\uparrow }&= A(2k)+\dots + 2\zeta^{k(2k-1)},\\
(\psi_{2k+1})_{\uparrow\downarrow\cdots\uparrow\downarrow\uparrow }&= A(2k+1)+\dots + \zeta^{k(2k+1)},
\end{align}
where $\cdots$ denotes intermediate powers of $\zeta$.
\end{proposition}
\begin{proof}
We compute the coefficients of the lowest-order term through the evaluation of the components at $\zeta=0$. To this end, we use \cite{zinnjustin:13}
    \begin{equation}
    \label{eqn:TrigLimitH2kH2kJ2}
	\left.H_{2k}\right|_{\zeta=0} = A_{\textup{\tiny V}}(2k+1) 
	\, , \quad
	\left. H_{2k}(J_2)\right|_{\zeta=0} = 2^{1-k}\frac{A(2k-1)}{A_{\textup{\tiny V}}(2k-1)}.
	\end{equation}
	Hence, we obtain
\begin{align}
  \left.(\psi_{2k+1})_{\uparrow\downarrow\cdots\uparrow\downarrow\uparrow }\right|_{\zeta=0}&=2^k\left.H_{2k}\right|_{\zeta=0}\left.H_{2(k+1)}(J_2)\right|_{\zeta=0} = A(2k+1),\\
   \left.(\psi_{2k})_{\uparrow\downarrow\cdots\uparrow\downarrow\uparrow }\right|_{\zeta=0}&=2^k\left.H_{2k}\right|_{\zeta=0}\left.H_{2k}(J_2)\right|_{\zeta=0} = \frac{2A_{\textup{\tiny V}}(2k+1)A(2k-1)}{A_{\textup{\tiny V}}(2k-1)}=A(2k),
\end{align}
where the last equality of the second line follows from Corollary 21 of \cite{razumov:04_2}. 

To find the coefficients of the highest-order terms of the components, we analyse the highest-order terms of $H_{2k}(w,w')$ for fixed $w,w'$. For $k=1$ this analysis is trivial, since $H_2(w,w')=1$. For $k\geqslant 2$, we use the determinant formula of \cref{lem:HDetFormula2}, and find
\begin{equation}
   H_{2k}(w,w') =  (1+w+w')(1+w)^{k-2}(1+w')^{k-2}
        \zeta^{k(k-1)} +  \dots,
\end{equation}
where $\dots$ denotes lower-order terms in $\zeta$.
Hence, we have
\begin{equation}
  H_{2k}= \zeta^{k(k-1)}+\dots, \quad H_{2k}(J_2)= 2^{1-k}\zeta^{k(k-1)}+\dots
\end{equation}
Using these results, the highest-order terms of the components follow from the explicit expressions given above.
\end{proof}

It follows from \cref{prop:AlternatingComponent,prop:AltComponentComb} that the components $(\psi_{n})_{\uparrow\downarrow\cdots\uparrow\downarrow\uparrow }$ and $(\bar \psi_{n})_{\uparrow\downarrow\cdots\uparrow\downarrow\uparrow } = (-1)^{n+1} (\psi_{n})_{\uparrow\downarrow\cdots\uparrow\downarrow\uparrow }$ do not identically vanish. Hence, the vectors $|\psi_n\rangle$, $|\bar \psi_n\rangle$ do not identically vanish,
as was anticipated above.
Their linear independence follows from the fact that they are, by construction, eigenvectors of the spin-reversal operator with different eigenvalues: $F|\psi_n\rangle = (-1)^n|\psi_n\rangle$, $F |\bar \psi_n\rangle = (-1)^{n+1}|\bar \psi_n\rangle$.

\subsubsection{The trigonometric limit}

We now evaluate the scalar product $S_n$ for $\zeta \to 0$. This evaluation corresponds to the limit $p\to 0$. After a rescaling, the weights of the eight-vertex model tend, in this limit, to the trigonometric weights of the six-vertex model at $\eta=\pi/3$. Hence, we refer to it as the \textit{trigonometric limit}. The trigonometric limit of $S_n$ was rigorously computed in \cite{razumov:07} with the help of contour-integral formulas for the components of the  ground-state vectors of the XXZ chain at $\Delta=-1/2$. Here, we show that it also follows from \cref{thm:S} and a few properties of Schur functions and symplectic characters.

Let $k\geqslant 1$ and $\lambda=(\lambda_1,\dots,\lambda_k)$ be a partition. We recall that the Schur function $s_\lambda$ and the symplectic character $\chi_\lambda$ associated with $\lambda$ are given by
\begin{equation}
  s_{\lambda}(z_1,\dots,z_k) = \frac{\det_{i,j=1}^{k}\left(z_i^{\lambda_j+k-j}\right)}{\det_{i,j=1}^{k}\left(z_i^{k-j}\right)},
\end{equation}
and
\begin{equation}
\chi_{\lambda}(z_1,\dots,z_m) = \frac{\det_{i,j=1}^{k}\left(z_i^{\lambda_j+k-j+1}-z_i^{-(\lambda_j+k-j+1)}\right)}{\det_{i,j=1}^{k}\left(z_i^{k-j+1}-z_i^{-(k-j+1)}\right)},
\end{equation}
respectively. Hereafter, we focus on the case where $\lambda$ is given by the double-staircase partition $Y_k = (\lfloor (k-i)/2\rfloor)_{i=1}^{k}$. For $\zeta=0$, the polynomial $H_{2k}(w_1,\dots,w_{2k})$ is related to a symplectic character associated with this partition \cite{zinnjustin:13}. Through the parameterisation
\begin{equation}
\bar w(z) =
\frac{(z-1)^2}{1+z+z^2},
\end{equation}
one obtains the relation
\begin{equation}
\label{eqn:TrigRosPolChi}
	\left.H_{2k}(\bar w(z_1),\dots,\bar w(z_{2k}))\right|_{\zeta=0}
	=
	3^{k(k-1)}
	\prod_{i=1}^{2k}\left(1+z_i+z_i^{-1}\right)^{-k+1}	\chi_{Y_{2k}}(z_1,\dots,z_{2k}),
\end{equation}
for each $k\geqslant 1$. 

The next lemma provides particular factorisation properties of Schur functions associated with the double-staircase partitions into symplectic characters. It can be proven through elementary row and column operations in the involved determinants \cite{ayyer:19}.
\begin{lemma}
  \label{lem:SchurFactorisation}
	Let $\omega=\ee^{\i \pi/3}$ and $k\geqslant 0$ be an integer, then we have
	\begin{multline}
	s_{Y_{4k+2}}(z_1,\dots,z_{2k},z_1^{-1},\dots,z_{2k}^{-1},z,1)
	=
	z^k
	\prod_{i=1}^{2k}\left(1+z_i+z_i^{-1}\right) \\
	\times
	\chi_{Y_{2k+2}}(z_1,\dots,z_{2k},z,\omega)
	\chi_{Y_{2k}}(z_1,\dots,z_{2k}),
	\end{multline}
	and
	\begin{multline}
	s_{Y_{4k+4}}(z_1,\dots,z_{2k+1},z_1^{-1},\dots,z_{2k+1}^{-1},z,1)
	=
	z^k(1+z)
	\prod_{i=1}^{2k+1}\left(1+z_i+z_i^{-1}\right) \\
	\times
	\chi_{Y_{2k+2}}(z_1,\dots,z_{2k+1},z)
	\chi_{Y_{2k+2}}(z_1,\dots,z_{2k+1},\omega).
	\end{multline}
\end{lemma}

We now compute the trigonometric limit of $S_n$. To this end, we recall that an alternating sign matrix of size $n$ has a unique $+1$ in its first row. Let $A(n,k)$ denote the number of the matrices for which this $+1$ is in column $k$. This number is given by \cite{bressoudbook}
\begin{equation}
  A(n,k) = \frac{\binom{n+k-2}{ n-1} \binom{2n-1-k}{n-1}}{\binom{3n-2}{n-1}}A(n).
\end{equation}
\begin{proposition}
	\label{prop:STrig}
	We have
	\begin{equation}
	\left. S_n \right|_{\zeta=0}
	=
	\sum_{k=0}^{n}
	A(n+1,k+1)
	\mu^k.
	\label{eqn:STrig}
	\end{equation}
\end{proposition}
\begin{proof}
We evaluate the expressions of \cref{thm:S} for $\zeta=0$ with the help of \eqref{eqn:TrigRosPolChi}. This evaluation yields
	\begin{align}
	\left.S_{2k}\right|_{\zeta=0}
	&=
	3^{-k(2k-1)}(\mu^2-\mu+1)^k
	\chi_{Y_{2k}}(1,\dots,1)
	\chi_{Y_{2(k+1)}}(1,\dots,1,\omega,z(\mu)),\label{eqn:TrigSChiEven}\\
	\left.S_{2k+1}\right|_{\zeta=0}
	&=
	3^{-k(2k+1)}
	(\mu+1)
	(\mu^2-\mu+1)^k
	\chi_{Y_{2(k+1)}}(1,\dots,1,\omega)\chi_{Y_{2(k+1)}}(1,\dots,1,z(\mu)),
	\label{eqn:TrigSChiOdd}
	\end{align}
	for $n=2k$ and $n=2k+1$, respectively. Here, $z(\mu)=(\omega - \mu)/(\mu\omega -1)$. 

    Next, we apply \cref{lem:SchurFactorisation} and rewrite the products of symplectic characters in terms of Schur functions. For both $n=2k$ and $n=2k+1$, we obtain
	\begin{equation}
	\label{eqn:SnZeta0Int}
	  S_n|_{\zeta=0}=3^{-n(n+1)/2}(\omega(1-\omega \mu))^{n}s_{Y_{2(n+1)}}(1,\dots,1,z(\mu)).
	\end{equation}
	
	Finally, we note that that the Schur function on the right-hand side is a specialisation of the partition function of a six-vertex model on an $(n+1)\times (n+1)$ square grid with domain-wall boundary conditions \cite{okada:06}. The model's configurations are in bijection with the set of alternating sign matrices of size $n+1$. This bijection leads to \cite{bressoudbook}
	\begin{equation}
	  s_{Y_{2(n+1)}}(1,\dots,1,z(\mu)) = \frac{3^{n(n+1)/2}}{(\omega(1-\omega \mu))^{n}}\sum_{k=0}^{n}A(n+1,k+1)\mu^k.
	\end{equation}
	The substitution of this expression into \eqref{eqn:SnZeta0Int} ends the proof.
\end{proof}

\subsection{The homogeneous limit of \texorpdfstring{$\boldsymbol{\bar Z_n^\pm}$}{$\bar Z_n^\pm$}}
\label{sec:HomLimitZBar}

In this section, we compute the scalar products 
\begin{equation}
  \label{eqn:DefSBar}
  \bar S_n^\pm  = \left(\left(\langle{\uparrow\uparrow}|+\nu\langle{\downarrow\downarrow}|\right)^{\otimes n}\otimes(\langle{\uparrow}|\pm\langle{\downarrow}|)\right)|\psi_{n}\rangle
\end{equation}
from the homogeneous limit of $\bar Z_n^\pm$. As for the previous section, we divide this one into three parts. First, we find closed-form expressions for $\bar S_n^\pm$. Second, we use them to compute and analyse the components  $(\psi_n)_{\downarrow\cdots \downarrow \downarrow}$ and $ (\psi_n)_{\uparrow\cdots \uparrow \downarrow}$. 
Third, we evaluate the trigonometric limit of $\bar S_n^\pm$.

\subsubsection{Closed-form expression}
We use the notation
\begin{equation}
    \bar \nu = \frac{(\nu-\zeta)(\nu\zeta-1)}{(\zeta^2-1)\nu}.     
  \end{equation}
Furthermore, for each $\epsilon,\epsilon'=\pm$, we define
\begin{equation}
  \label{eqn:DefCD}
  c_{\epsilon\epsilon'} =(\zeta+\epsilon)(\nu+\epsilon'), \quad d_\epsilon = \nu(\zeta+\epsilon)^2, \quad \bar d_\epsilon = \zeta(\nu+\epsilon)^2.
\end{equation}
\begin{theorem}
\label{thm:SBar}
  For each $k\geqslant 1$, the scalar products $\bar S_{2k}^\pm$ and $\bar S_{2k+1}^\pm$ are polynomials in $\nu$ and $\zeta$, given by
  \begin{align}
    \bar S_{2k}^+ &= \frac{\nu^k\left(c_{+-}^2H_{2(k+1)}(J_4)H_{2(k+1)}(J_3,\bar \nu)-c_{-+}^2H_{2(k+1)}(J_3)H_{2(k+1)}(J_4,\bar \nu)\right)}{2(\nu-\zeta)(\nu \zeta-1)},\\
    \bar S_{2k}^- &= \frac{c_{+-}c_{-+}\nu^k\left(H_{2(k+1)}(J_4)H_{2(k+1)}(J_3,\bar \nu)-H_{2(k+1)}(J_3)H_{2(k+1)}(J_4,\bar \nu)\right)}{2(\nu-\zeta)(\nu \zeta-1)},
  \end{align}
  and
  \begin{align}
    \bar S_{2k+1}^+ &= \frac{c_{--}\nu^k\left(d_+H_{2(k+1)}H_{2(k+2)}(J_3,J_4,\bar \nu)-\bar d_+H_{2(k+1)}(\bar \nu)H_{2(k+2)}(J_3,J_4)\right)}{2\zeta(\nu-\zeta)(\nu \zeta-1)},\\
    \bar S_{2k+1}^- &= \frac{c_{++}\nu^k\left(d_-H_{2(k+1)}H_{2(k+2)}(J_3,J_4,\bar \nu){-}\bar d_{-}H_{2(k+1)}(\bar \nu)H_{2(k+2)}(J_3,J_4)\right)}{2\zeta(\nu-\zeta)(\nu \zeta-1)}.\end{align}
  \end{theorem}

\begin{proof} The proof is similar to the proof of \cref{thm:S}. We divide it into two parts. The first part consists of establishing the explicit formulas for the scalar products. In the second part, we show that they define polynomials in $\nu$ and $\zeta$.
\medskip

\textit{Part 1: Explicit formulas.}
First, we evaluate the scalar products $\bar Z_n^\pm$ for $u_1=\dots=u_{2n+1}=0$. Using \eqref{eqn:DefZBar}, we find
\begin{equation}
  \label{eqn:ZBarHomLimit}
  Z_n^\pm(0,\dots,0) = (\vartheta_1(\lambda,p^2)\vartheta_1(\lambda-\eta,p^2))^n\left((\langle{\uparrow\uparrow}|+\nu\langle{\downarrow\downarrow}|)^{\otimes n}\otimes(\langle{\uparrow}|\pm \langle{\downarrow}|)\right) |\Psi_n(0,\dots,0)\rangle,
\end{equation}
where
\begin{equation}
\label{eqn:DefNu}
\nu = \frac{\vartheta_4(\lambda-\eta,p^2)\vartheta_4(\lambda,p^2)}{\vartheta_1(\lambda-\eta,p^2)\vartheta_1(\lambda,p^2)}.
\end{equation}
We use \cref{prop:ZX} to write the left-hand side of \eqref{eqn:ZBarHomLimit} in terms of $\bar X^\pm_n$. Moreover, we use \eqref{eqn:HomLimitPsi} and \eqref{eqn:DefSBar} to write the right-hand side in terms of $\bar S_n^\pm$. Solving for the latter, we obtain
\begin{equation}
   \bar S_n^\pm=\left(\frac{\vartheta_1(2\eta,p^2)}{\vartheta_1(\lambda,p^2)\vartheta_1(\lambda-\eta,p^2)}\right)^n \mathcal N_n\bar X^\pm_n(0,\dots,0).
\end{equation}

Second, we replace $\mathcal N_n$ by its definition and $\bar X_{n}^\pm$ by its explicit form in terms of the elliptic Tsuchiya determinant, given in \cref{thm:XBar}. This step requires to consider the cases $n=2k$ and $n=2k+1$ separately. As in the proof of \cref{thm:S}, we present the details for $n=2k$:

\begin{multline}
   \bar S_{2k}^\pm=\frac{p^k\ee^{-4k\i\eta}}{\vartheta_1(\eta)^{4k^2}}\left(\frac{\vartheta_1(\eta,p^2)}{\vartheta_1(\lambda,p^2)\vartheta_1(\lambda-\eta,p^2)\vartheta_4(0,p^2)}\right)^{2k} \left(\frac{\vartheta_4(0,p^2)}{\vartheta_4(\eta,p^2)}\right)^{k(2k+1)} \\
   \times\Bigl(\gamma_0^\pm \mathbb H_{2(k+1)}(0,\dots,0,\beta_4)\mathbb H_{2(k+1)}(0,\dots,0,w(\eta+\lambda),\beta_3)\\+\delta_0^\pm \mathbb H_{2(k+1)}(0,\dots,0,\beta_3)\mathbb H_{2(k+1)}(0,\dots,0,w(\eta+\lambda),\beta_4)\Bigr).
\end{multline}  
We now rewrite the elliptic Tsuchiya determinants in terms of the polynomials of \cref{sec:RGPolynomials}, using \cref{lem:UniformisationH}, and find
\begin{multline}
  \bar S_{2k}^\pm=\nu^k\bigl(\gamma_0^\pm H_{2(k+1)}(w(\beta_4))H_{2(k+1)}(w(\beta_3),w(\lambda+\eta))\\
  +\delta_0^\pm H_{2(k+1)}(w(\beta_3))H_{2(k+1)}(w(\beta_4),w(\lambda+\eta))\bigr).
\end{multline}
We note that
\begin{equation}
 w(\beta_3)=J_3,\quad w(\beta_4)=J_4,\quad w(\lambda+\eta) = \frac{(\nu - \zeta)(\zeta \nu -1)}{(\zeta^2-1)\nu} = \bar \nu,
\end{equation}
and
\begin{align}
  \gamma_0^+ &= \frac{c_{+-}^2}{2(\nu-\zeta)(\nu\zeta-1)},\quad \delta_0^+ = - \frac{c_{-+}^2}{2(\nu-\zeta)(\nu\zeta-1)},\\
   \gamma_0^-& =\frac{c_{+-}c_{-+}}{2(\nu-\zeta)(\nu\zeta-1)},\quad \delta_0^-= -\frac{c_{+-}c_{-+}}{2(\nu-\zeta)(\nu\zeta-1)}.
\end{align}
These relations lead to the expressions for $\bar S_{2k}^\pm$ given above. The derivation of the expressions for $\bar S_{2k+1}^\pm$ is similar.
\medskip

\textit{Part 2: Polynomial nature.} We now show that $S_n^\pm$ is a polynomial in $\nu$ and $\zeta$, focussing on $n=2k$. Using \cref{lem:HDetFormula2}, it is straightforward to show that $H_{2(k+1)}(J_i)$ and $\nu^k H_{2(k+1)}(J_i,\bar \nu)$ are polynomials in $\nu$ and $\zeta$ for both $i=3$ and $i=4$. Hence, $\bar S_{2k}^\pm$ is a ratio of polynomials in $\nu$ and $\zeta$.  Its denominator tends to zero if and only if $\nu \to \zeta$ or $\nu\zeta\to 1$. The numerator tends to zero in these limits, too.
This is sufficient to conclude that the ratio is a polynomial in both $\nu$ and $\zeta$. The case $n=2k+1$ is similar albeit slightly more technical.

\medskip
Finally, we note that this proof relies on the meromorphic parameterisation \eqref{eqn:DefNu} of $\nu$ in terms of $\lambda$. Similarly to \cref{thm:S}, one checks that, under this parameterisation, any real $\nu$ has a point of its preimage in $0 < \lambda < \pi$. Hence, the theorem holds for all real $\nu$, and, by analytic continuation, for all complex $\nu$.
\end{proof}

\subsubsection{Components}
\begin{proposition}
\label{prop:APComponent}
  For each $k\geqslant 0$, we have
\begin{align}
  \label{eqn:PE1}
  (\psi_{2k})_{\downarrow\cdots \downarrow\downarrow} &= \zeta^k H_{2k}H_{2(k+1)}(J_3,J_4),\\
   \label{eqn:PE2}
(\psi_{2k+1})_{\downarrow\cdots \downarrow\downarrow} &= \zeta^{k+1} H_{2(k+1)}(J_3)H_{2(k+1)}(J_4),
\end{align}
and
\begin{align}
  \label{eqn:APE1}
(\psi_{2k})_{\uparrow\cdots\uparrow\downarrow} &= \frac{1}{2}\zeta^k \left((1+\zeta)H_{2k}(J_3)H_{2(k+1)}(J_4)+(1-\zeta)H_{2k}(J_4)H_{2(k+1)}(J_3)\right),\\
  \label{eqn:APE2}
(\psi_{2k+1})_{\uparrow\cdots\uparrow\downarrow} &= \frac{1}{2}\zeta^k \left((1-\zeta^2)H_{2(k+1)}H_{2(k+1)}(J_3,J_4)+H_{2k}H_{2(k+2)}(J_3,J_4)\right).
\end{align}
\end{proposition}
\begin{proof}
In terms of the scalar products $\bar S_n^\pm$, we have
\begin{equation}
  (\psi_{n})_{\downarrow\cdots \downarrow\downarrow} = \frac{(-1)^n}{2}(\left.\bar S_{n}^+\right|_{\nu=0}+\left.\bar S_{n}^-\right|_{\nu=0}),\quad 
  (\psi_{n})_{\uparrow\cdots\uparrow\downarrow} =\frac12 \left(\left.\bar S_n^+\right|_{\nu=0}-\left.\bar S_n^-\right|_{\nu=0}\right),
\end{equation}
where we used $(\psi_{n})_{\downarrow\cdots \downarrow\downarrow} =(-1)^n(\psi_{n})_{\uparrow\cdots \uparrow\uparrow}$ to obtain the first expression. We compute the scalar products using the explicit expressions of \cref{thm:SBar}. With the help of \cref{lem:HInfinite}, we obtain
\begin{align}
  \left.\bar S_{2k}^+\right|_{\nu=0} &= \frac{\zeta^{k-1}}{2}\left((1+\zeta)^2H_{2k}(J_3)H_{2(k+1)}(J_4)-(1-\zeta)^2H_{2k}(J_4)H_{2(k+1)}(J_3)\right),\\
  \left.\bar S_{2k}^-\right|_{\nu=0} &=\frac{\zeta^{k-1}}{2}(1-\zeta^2)\left(H_{2k}(J_3)H_{2(k+1)}(J_4)-H_{2k}(J_4)H_{2(k+1)}(J_3)\right),
\end{align}
and
 \begin{align}
  \left.\bar S_{2k+1}^+\right|_{\nu=0} &=\frac{\zeta^{k-1}(1-\zeta)}{2}\left((1+\zeta)^2H_{2(k+1)}H_{2(k+1)}(J_3,J_4)-H_{2k}H_{2(k+2)}(J_3,J_4)\right),\\
  \left.\bar S_{2k+1}^-\right|_{\nu = 0} &=\frac{\zeta^{k-1}(1+\zeta)}{2}\left((1-\zeta)^2H_{2(k+1)}H_{2(k+1)}(J_3,J_4)-H_{2k}H_{2(k+2)}(J_3,J_4)\right).
\end{align}

For the components labelled by the polarised spin configuration, we thus find
\begin{align}
  (\psi_{2k})_{\downarrow\cdots\downarrow\downarrow} &= \frac{\zeta^{k-1}}{2}\left((\zeta+1)H_{2k}(J_3)H_{2(k+1)}(J_4)+(\zeta-1)H_{2k}(J_4)H_{2(k+1)}(J_3)\right),\\
  (\psi_{2k+1})_{\downarrow\cdots\downarrow\downarrow} &= \frac{\zeta^{k-1}}{2}\left(H_{2k}H_{2(k+2)}(J_3,J_4)-(1-\zeta^2)H_{2(k+1)}H_{2(k+1)}(J_3,J_4)\right),
\end{align}
for $n=2k$ and $n=2k+1$, respectively. We simplify these expressions with the help of the bilinear identities of \cref{lem:Bilinear}. Indeed, we rewrite the first line by using \eqref{eqn:Bilinear1} with $w_1,\dots,w_{2k-1}=0$ and $x=J_3,\,y=J_4,\,u=v=0$. To simplify the second line, we use the identity \eqref{eqn:Bilinear2}, specialised to $w_1,\dots,w_{2k}=0$ and $x=J_3,\,y=J_4,\,u=v=0$. These simplifications lead to the expressions \eqref{eqn:PE1} and \eqref{eqn:PE2}.

For the components labelled by the almost-polarised spin configuration, the substitution yields \eqref{eqn:APE1} and \eqref{eqn:APE2}. 
\end{proof}

It follows that the components  $(\psi_n)_{\downarrow\cdots\downarrow\downarrow}$ and $(\psi_n)_{\uparrow\cdots\uparrow\downarrow}$ are polynomials in $\zeta$ with integer coefficients. In the next proposition, we provide the coefficients of their lowest-order and highest-order terms. They follow from the evaluations \cite{zinnjustin:13}
\begin{equation}
	\left. H_{2(k+1)}(J_3,J_4) \right|_{\zeta=0} = A_{\textup{\tiny V}}(2k+1),  \quad
	\left. H_{2k}(J_{3})\right|_{\zeta=0} =\left. H_{2k}(J_{4})\right|_{\zeta=0} = N_8(2k),
	\label{eqn:CombAVN8}
\end{equation}
and from \cref{lem:HDetFormula2}, respectively.

\begin{proposition}
\label{prop:APComponentComb}
For each $k\geqslant 0$, we have
\begin{align}
  (\psi_{2k})_{\downarrow\cdots\downarrow\downarrow} &= A_{\textup{\tiny V}}(2k+1)^2\zeta^k+\dots +\zeta^{k(2k+1)},\\
(\psi_{2k+1})_{\downarrow\cdots \downarrow\downarrow} &= N_8(2k+2)^2\zeta^{k+1} +\cdots +\zeta^{(k+1)(2k+1)},
\end{align}
and, where $\cdots$ denotes intermediate powers of $\zeta$, and 
\begin{align}
  (\psi_{2k})_{\uparrow\cdots\uparrow\downarrow} &= N_8(2k)N_8(2k+2)\zeta^k  +\cdots, \\
(\psi_{2k+1})_{\uparrow\cdots\uparrow\downarrow} &=  A_{\textup{\tiny V}}(2k+1)A_{\textup{\tiny V}}(2k+3)\zeta^k+\cdots,
\end{align}
where $\cdots$ denotes higher powers of $\zeta$.
\end{proposition}
The coefficient of the highest-order term of $(\psi_{n})_{\uparrow\cdots\uparrow\downarrow}$ is rather difficult to compute. The reason is a nontrivial cancellation between several of the leading powers of the two terms on the right-hand sides of \eqref{eqn:APE1} and \eqref{eqn:APE2}.
For small $n$, we observe that the coefficient of the highest-order term is given by the $n$-th Catalan number $C(n) = \frac{1}{n+1}\binom{2n}{n}$, but we have no proof for arbitrary $n$.

\subsubsection{The trigonometric limit}
Finally, we compute the trigonometric limit of $\bar S_n^\pm$ . It yields simple polynomials in $\nu$, in sharp contrast to the trigonometric limit of $S_n$ discussed above. Despite their simplicity, we have not found them in the literature on the XXZ spin chain at $\Delta=-1/2$.
\begin{proposition}
	\label{prop:SBarTrig}
	For each $k\geqslant 0$, we have
		\begin{equation}
		\label{eqn:SBar2kZeta0}
		\left.\bar S^+_{2k}\right|_{\zeta=0} = 2A_{\textup{\tiny V}}(2k+1)N_8(2(k+1))\nu^k,\quad \left.\bar S^-_{2k}\right|_{\zeta=0} = 0,
		\end{equation}
		and
		\begin{equation}
		\label{eqn:SBar2k1Zeta0}
		\left.\bar S^\pm_{2k+1}\right|_{\zeta=0} 
		= -A_{\textup{\tiny V}}(2k+3)N_8(2(k+1))( \nu \mp 1)\nu^k.
		\end{equation}
\end{proposition}

\begin{proof}
	We use \eqref{eqn:CombAVN8} to evaluate the closed-form expressions found in \cref{thm:SBar} for $\zeta=0$. The evaluation of $\bar S_{2k}^\pm$ is straightforward and leads to \eqref{eqn:SBar2kZeta0}. The evaluation of $\bar{S}^\pm_{2k+1}$ leads, however, to a singular expression. To compute it, we use the bilinear identity \eqref{eqn:Bilinear1} of \cref{lem:Bilinear} with $w_1=\dots=w_{2k-1}=0$. Setting $x=J_3,y=J_4,u=\bar{\nu},v=0$, we obtain
\begin{multline}
  H_{2(k+1)}(J_3,J_4)H_{2k}(\bar{\nu})=\frac{1}{2\nu} \bigl((1+(1-\zeta)\nu + \nu^2)H_{2(k+1)}(J_4,\bar{\nu})H_{2k}(J_3)\\
  -(1-(1+\zeta)\nu + \nu^2)H_{2(k+1)}(J_3,\bar{\nu})H_{2k}(J_4)\bigr).
\end{multline}
Likewise, setting $x=J_3,y=J_4,u=0,v=\bar{\nu}$, we find
\begin{multline}
  H_{2(k+1)}(J_3,J_4,\bar{\nu})H_{2k}=\frac{\zeta}{2\nu^2(\zeta^2-1)}
		\bigl((\nu-1)^2(1+(1-\zeta)\nu + \nu^2)H_{2(k+1)}(J_3,\bar{\nu})H_{2k}(J_4)\\-
		 (\nu+1)^2(1-(1+\zeta)\nu + \nu^2)H_{2(k+1)}(J_4,\bar{\nu})H_{2k}(J_3) \bigr).
\end{multline}
Applying these relations allows us to simplify the scalar products to
\begin{multline}
    \bar S^+_{2k+1} = \frac{\nu^{k-1}(\nu-1)}{2}\bigr((1+\nu^2)H_{2(k+1)}(J_4)H_{2(k+2)}(J_3,\bar \nu)\\-(1+\nu)^2H_{2(k+1)}(J_3)H_{2(k+2)}(J_4,\bar \nu)\bigl),
    \label{SbarOdd+Re}
  \end{multline}
  and
  \begin{multline}
    \bar S^-_{2k+1} = \frac{\nu^{k-1}(\nu+1)}{2}\bigr((1-\nu)^2 H_{2(k+1)}(J_4)H_{2(k+2)}(J_3,\bar \nu)\\-(1+\nu^2)H_{2(k+1)}(J_3)H_{2(k+2)}(J_4,\bar \nu)\bigl).
    \label{SbarOdd-Re}
  \end{multline}
  These expressions are non-singular for $\zeta=0$. We evaluate them with the help of  \eqref{eqn:CombAVN8}, which leads to \eqref{eqn:SBar2k1Zeta0}.	
\end{proof}

\subsection{The sums of components}
\label{sec:SumOfComponents}
In this section, we compute the sum of the components of the vectors $|\psi_n\rangle$ and $|\bar \psi_n\rangle$, given by
\begin{equation}
  \label{eqn:DefSigma}
  \Sigma_n = \sum_{\bm \alpha} (\psi_n)_{\bm \alpha}, \quad \bar\Sigma_n  = \sum_{\bm \alpha} (\bar\psi_n)_{\bm \alpha},
\end{equation}
respectively. It is possible to find closed-form expressions for these sums from the homogeneous limit of a scalar product for the inhomogeneous supersymmetric eight-vertex model. To find these scalar products, one needs to follow the strategy of \cref{sec:ScalarProducts}, but with a different solution to the boundary Yang-Baxter equation \cite{brasseur:19}. 
Here, we present a different approach that exploits the relation between the eight-vertex model and the XYZ spin chain.

\subsubsection{XYZ spin chain and parameter range}

Let us define the XYZ Hamiltonian
\begin{equation}
  \label{eqn:HamXYZ}
  H = -\frac{1}{2} \sum_{i=1}^L J_4\sigma_i^x\sigma_{i+1}^x+J_3\sigma_i^y\sigma_{i+1}^y+J_2\sigma_i^z\sigma_{i+1}^z,
\end{equation}
whose coupling constants are given in \eqref{eqn:DefJ}. Furthermore, we introduce
\begin{equation}
  \label{eqn:E0}
  E_0 =-\frac{(2n+1)(3+\zeta^2)}{4(1-\zeta^2)}.
\end{equation}
The vectors $|\psi_n\rangle$ and $|\bar \psi_n\rangle$ span the one-dimensional spaces of solutions of the eigenvalue problems
\begin{align}
    \label{eqn:HSystem}
  H|\psi\rangle &= E_0|\psi\rangle, \quad F|\psi\rangle = (-1)^n |\psi\rangle,\\
  H|\psi\rangle &= E_0|\psi\rangle, \quad F|\psi\rangle = (-1)^{n+1} |\psi\rangle,
\end{align}
where $|\psi\rangle \in V^{2n+1}$, respectively \cite{hagendorf:18}. We note that the Hamiltonian's eigenvalue $E_0$ is doubly degenerate. For $0<\zeta<1$ (and even $-1< \zeta<1$), it is its ground-state eigenvalue. 

We now use the connection to the spin chain to characterise the vectors as functions of $\zeta$:
\begin{lemma}
  \label{lem:RationalPsi}
  The vectors $|\psi_n\rangle$ and $|\bar \psi_n\rangle$ are rational functions of $\zeta$.
\end{lemma}
\begin{proof}
  The space of the solution of the eigenvalue problem \eqref{eqn:HSystem} is one-dimensional. Hence, if one fixes one component of $|\psi\rangle$ then all other components follow from basic linear algebra. For the solution $|\psi_n\rangle$, we have shown in \cref{prop:AlternatingComponent,prop:AltComponentComb} that the component $(\psi_n)_{\uparrow\downarrow\cdots\uparrow\downarrow\uparrow}$ is a non-vanishing polynomial in $\zeta$. Since both $H$ and $E_0$ are rational functions of $\zeta$, all the components of $|\psi_n\rangle$ are therefore rational in $\zeta$. The rationality of $|\bar \psi_n\rangle$ follows from the relation $|\bar \psi_n\rangle = P|\psi_n\rangle$.
\end{proof}

This lemma allows us to consider the vectors $|\psi_n\rangle$ and $|\bar \psi_n\rangle$ not only as rational functions on the interval $0<\zeta<1$ but on the the real line, forgetting about the initial parameterisation \eqref{eqn:ZetaP}. We adopt this point of view in the following and note that the vectors span the eigenspace of $E_0$ for all $\zeta$. The following lemma gives the behaviour of $|\psi_n\rangle$ for large $\zeta$:
\begin{lemma}
\label{lem:PsiZetaInfinity}
We have $|\psi_n\rangle = \zeta^{n(n+1)/2}\left(|{\downarrow\cdots\downarrow}\rangle + (-1)^n |{\uparrow\cdots\uparrow}\rangle+ o(1)\right)$ as $\zeta \to \infty$.
\end{lemma}

\begin{proof}
  For $\zeta \to \infty$, the eigenvalue problem \eqref{eqn:HSystem} becomes
  \begin{equation}
    \frac{1}{4}\left(\sum_{i=1}^{2n+1}\sigma_i^z\sigma_{i+1}^z\right)|\psi\rangle = \frac{2n+1}{4}|\psi\rangle, \quad F|\psi\rangle = (-1)^n|\psi\rangle.
  \end{equation}
  The space of its solutions is spanned by $|\psi\rangle =|{\downarrow\cdots\downarrow}\rangle + (-1)^n |{\uparrow\cdots\uparrow}\rangle$. By \cref{lem:RationalPsi}, $|\psi_n\rangle$ is a rational function of $\zeta$. We conclude that there are an integer $d_n$ and a complex number $a_n$ such that
  \begin{equation}
    |\psi_n\rangle = a_n\zeta^{d_n}\left(|{\downarrow\cdots\downarrow}\rangle + (-1)^n |{\uparrow\cdots\uparrow}\rangle+ o(1)\right),
  \end{equation}
  It follows from \cref{prop:APComponentComb} that $a_n=1$ and $d_n=n(n+1)/2$.
\end{proof}

\subsubsection{The sums of components}

We now compute $\Sigma_n$ and $\bar \Sigma_n$. Using the spin-reversal properties $F|\psi_n\rangle = (-1)^n|\psi_n\rangle$ and $F|\bar \psi_n\rangle = (-1)^{n+1}|\bar\psi_n\rangle$, we find the trivial results
\begin{equation}
  \Sigma_{2k+1} = 0, \quad \bar \Sigma_{2k} = 0,
  \label{eqn:SigmaTrivialCases}
\end{equation}
for each $k\geqslant 0$. We are going to show that $\Sigma_{2k}$ and $\bar \Sigma_{2k+1}$ are, however, quite nontrivial. To this end, it will be useful to write them as follows:
\begin{equation}
  \label{eqn:SigmaU}
  \Sigma_n = 2^{n+1/2}\langle {\uparrow\cdots \uparrow}|U|\psi_n\rangle, \quad \bar \Sigma_n = 2^{n+1/2}\langle {\downarrow\cdots \downarrow}|U|\psi_n\rangle.
\end{equation}
Here, $ U = 2^{-L/2}\prod_{j=1}^L (1+\i \sigma_j^y)$ is an orthogonal operator on $V^L$. In the next lemma, we compute the action of $U$ on the vector $|\psi_n\rangle=|\psi_n(\zeta)\rangle$.
\begin{lemma}
\label{lem:UOnPsi}
For each $n\geqslant 0$, we have 
\begin{equation}
  U|\psi_{n}(\zeta)\rangle = B_{n}(\zeta)\left(|\psi_{n}(\zeta')\rangle +(-1)^{n+1} |\bar \psi_{n}(\zeta')\rangle\right),
  \label{eqn:ActionUPsi}
\end{equation}
where $\zeta' = (\zeta+3)/(\zeta-1)$, and
\begin{equation}
  \label{eqn:An}
  B_n(\zeta)^2 = \frac{\|\psi_n(\zeta)\|^2}{2\|\psi_n(\zeta')\|^2}.
\end{equation}
\end{lemma}
\begin{proof}
  First, we write $H=H(\zeta)$ and $E_0=E_0(\zeta)$ for the XYZ Hamiltonian \eqref{eqn:HamXYZ} and its special eigenvalue \eqref{eqn:E0}. One checks that they satisfy the relations
 \begin{equation}
    H(\zeta')U = \left(\frac{\zeta-1}{2}\right)UH(\zeta), \quad E_0(\zeta') = \left(\frac{\zeta-1}{2}\right)E_0(\zeta).
  \end{equation}
It follows from these relations that $U|\psi_n(\zeta)\rangle$ is an eigenvector of the Hamiltonian $H(\zeta')$ associated to the eigenvalue $E_0(\zeta')$. Since the corresponding eigenspace is spanned by $|\psi_n(\zeta')\rangle$ and $|\bar \psi_n(\zeta')\rangle$, we may write
\begin{equation}
  U|\psi_n(\zeta)\rangle = B_n(\zeta)|\psi_n(\zeta')\rangle + \bar B_n(\zeta) |\bar \psi_n(\zeta')\rangle,
\end{equation}
where $B_n(\zeta),\bar B_n(\zeta)$ are coefficients.

Second, we compute the scalar product of both sides of this equality with the basis vectors $|{\uparrow\cdots\uparrow}\rangle$ and $|{\downarrow\cdots\downarrow}\rangle$. Using \eqref{eqn:SigmaU}, we find
\begin{align}
  \label{eqn:SigmaIntermediate}
  \begin{split}
  \Sigma_n = 2^{n+1/2} \left(B_n(\zeta)\psi_n(\zeta')_{\uparrow\cdots\uparrow}+\bar B_n(\zeta)\bar \psi_n(\zeta')_{\uparrow\cdots\uparrow}\right) = 2^{n+1/2}( B_n(\zeta)-\bar B_n(\zeta))\psi_n(\zeta')_{\uparrow\cdots\uparrow},\\
  \bar\Sigma_n = 2^{n+1/2} \left(B_n(\zeta)\psi_n(\zeta')_{\uparrow\cdots\uparrow}+\bar B_n(\zeta)\bar \psi_n(\zeta')_{\uparrow\cdots\uparrow}\right) = 2^{n+1/2}(B_{n}(\zeta)+\bar B_{n}(\zeta))\psi_n(\zeta')_{\downarrow\cdots\downarrow}.
  \end{split}
\end{align}
Here, we applied the relations $(\bar \psi_n)_{\uparrow\cdots\uparrow}= -(\psi_n)_{\uparrow\cdots\uparrow},  (\bar \psi_n)_{\downarrow\cdots\downarrow}= (\psi_n)_{\downarrow\cdots\downarrow}$, which straightforwardly follow from $|\bar \psi_n\rangle = P|\psi_n\rangle$. We now combine \eqref{eqn:SigmaTrivialCases} with \eqref{eqn:SigmaIntermediate}, and use the fact that the components labelled by the polarised spin configurations do not identically vanish (which follows from \cref{prop:APComponentComb}). This leads to 
\begin{equation}
  \bar B_{2k}(\zeta)= -B_{2k}(\zeta), \quad \bar B_{2k+1}(\zeta) = B_{2k+1}(\zeta),
\end{equation}
for each $k\geqslant 0$. Hence, we obtain \eqref{eqn:ActionUPsi}.

Third, we compute the scalar product of each side of \eqref{eqn:ActionUPsi} with itself. Using the fact that $U$ is orthogonal, as well as $\langle \bar \psi_n(\zeta')|\psi_n(\zeta')\rangle=0$ and $\|\bar \psi_n(\zeta')\|^2=\|\psi_n(\zeta')\|^2$, we obtain \eqref{eqn:An}.
\end{proof}

It is clear from this lemma why we need the continuation of $|\psi_n\rangle$ to values of $\zeta$ outside the range $0< \zeta < 1$. Indeed, if $\zeta$ is in this range then $-\infty < \zeta' < -3$. Moreover, the lemma shows that to find an explicit formula for the sums of components we need the square norm of the ground-state vector $|\psi_n\rangle$. One of the main results of \cite{zinnjustin:13} are the expressions
\begin{align}
   \label{eqn:SquareNormPsi}
   \begin{split}
     ||\psi_{2k}||^2 & =2^{2k+1}H_{2k}H_{2(k+1)}(J_2,J_3)H_{2(k+1)}(J_3,J_4)H_{2(k+1)}(J_2,J_4),\\
    ||\psi_{2k+1}||^2 & =  2^{2(k+1)}H_{2(k+1)}(J_2)H_{2(k+1)}(J_3)H_{2(k+1)}(J_4)H_{2(k+2)}(J_2,J_3,J_4),
   \end{split}
\end{align}
for each $k\geqslant 0$. We use them to obtain the following result:
\begin{proposition}
\label{prop:SumOfComps}
For each $k\geqslant 0$, we have
\begin{align}
  \Sigma_{2k} &= 2^{k+1}(\zeta+3)^kH_{2k}H_{2(k+1)}(J_2,J_3),\\
  \bar \Sigma_{2k+1} &= 2^{k+1}(\zeta+3)^{k+1} H_{2(k+1)}(J_2)H_{2(k+1)}(J_3).
\end{align}
\end{proposition}
\begin{proof}
  The scalar product of \eqref{eqn:ActionUPsi} and the basis vectors $|{\uparrow\cdots \uparrow}\rangle$ and $|{\downarrow\cdots\downarrow}\rangle$ implies
  \begin{equation}
    \Sigma_{2k} = 2^{2k+3/2} B_{2k}(\zeta)\psi_{2k}(\zeta')_{\downarrow\cdots\downarrow}, \quad \bar \Sigma_{2k+1} = 2^{2k+5/2} B_{2k+1}(\zeta)\psi_{2k+1}(\zeta')_{\downarrow\cdots\downarrow}.
  \end{equation}
  To find these two expressions, we used $(\bar \psi_n)_{\uparrow\cdots\uparrow}= -(\psi_n)_{\uparrow\cdots\uparrow},  (\bar \psi_n)_{\downarrow\cdots\downarrow}= (\psi_n)_{\downarrow\cdots\downarrow}$, as well as $(\psi_{2k})_{\uparrow\cdots\uparrow}=(\psi_{2k})_{\downarrow\cdots\downarrow}$. We compute the coefficients $B_{2k}(\zeta)$ and $B_{2k+1}(\zeta)$ in these expressions by applying the transformation property of \cref{lem:HTransform1} to the square norms \eqref{eqn:SquareNormPsi}. We find
  \begin{equation}
    B_{2k}(\zeta) = \frac{b_{2k}}{2^{1/2}}\left(\frac{\zeta-1}{2}\right)^{k(2k+1)}, \quad B_{2k+1}(\zeta) = \frac{b_{2k+1}}{2^{1/2}}\left(\frac{\zeta-1}{2}\right)^{(k+1)(2k+1)},
  \end{equation}
  where $b_{2k}, b_{2k+1}=\pm 1$ are signs that we fix below. Moreover, using \cref{lem:HTransform1} and \cref{prop:APComponent}, we obtain
    \begin{align}
    \psi_{2k}(\zeta')_{\downarrow\cdots\downarrow} &= \left(\frac{\zeta+3}{\zeta-1}\right)^k \left(\frac{2}{\zeta-1}\right)^{2k^2}H_{2k}H_{2(k+1)}(J_2,J_3),\\
    \quad \psi_{2k+1}(\zeta')_{\downarrow\cdots\downarrow} &= \left(\frac{\zeta+3}{\zeta-1}\right)^{k+1} \left(\frac{2}{\zeta-1}\right)^{2k(k+1)}H_{2k}H_{2(k+1)}(J_2,J_3).
  \end{align}
  Hence, the sums of components are
  \begin{align}
  \label{eqn:Sigma2kIntermediate}
  \Sigma_{2k} &= b_{2k}2^{k+1}(\zeta+3)^k H_{2k}H_{2(k+1)}(J_2,J_3),\\
  \label{eqn:BarSigma2k1Intermediate}
  \bar \Sigma_{2k+1} &= b_{2k+1}2^{k+1}(\zeta+3)^{k+1} H_{2(k+1)}(J_2)H_{2(k+1)}(J_3).
\end{align}
  
  To find the signs $b_{2k}$ and $b_{2k+1}$, we consider the limit
   $\zeta \to \infty$.
   On the one hand, \cref{lem:HDetFormula2} implies
 \begin{align}
   H_{2k} &= \zeta^{k(k-1)}(1+o(1)), & H_{2k}(J_2) &= 2^{1-k}\zeta^{k(k-1)}(1+o(1)),\\
   H_{2k}(J_3) &= \zeta^{k(k-1)}(1+o(1)), & H_{2k}(J_2,J_3) &= 2^{1-k}\zeta^{k(k-1)}(1+o(1)),
 \end{align}
and therefore
  \begin{equation}
   \Sigma_{2k} = 2b_{2k} \zeta^{k(2k+1)}(1+o(1)), \quad \bar \Sigma_{2k+1} = 2b_{2k+1}\zeta^{(k+1)(2k+1)}(1+o(1)).
 \end{equation}
 On the other hand, it follows from \cref{lem:PsiZetaInfinity} that
 \begin{equation}
   \label{eqn:SigmaLargeZeta}
   \Sigma_{2k} = 2 \zeta^{k(2k+1)}(1+o(1)), \quad \bar \Sigma_{2k+1} = 2\zeta^{(k+1)(2k+1)}(1+o(1)).
 \end{equation}
 We compare the coefficients of the leading terms and conclude that $b_{2k}=b_{2k+1}=1$.\end{proof}

\Cref{prop:SumOfComps} implies that the sums of components are polynomials in the parameter $\zeta$ with integer coefficients. As for the components studied in \cref{sec:HomLimitZ,sec:HomLimitZBar}, we compute the coefficients of its lowest- and highest-order term. For the lowest-order term, we recall that the number of diagonally- and antidiagonally-symmetric alternating sign matrices of size $2n+1$ is given by \cite{behrend:17}
\begin{equation}
  A_{\textrm{DAD}}(2n+1) = \prod_{i=0}^{n}\frac{(3i)!}{(n+i)!}.
\end{equation}
\begin{proposition}
For each $k\geqslant 0$, we have
\begin{align}
  \Sigma_{2k} &= 2A_{\textup{DAD}}(4k+1) + \dots + 2\zeta^{k(2k+1)}, \\
  \bar \Sigma_{2k+1} &= 2A_{\textup{DAD}}(4k+3) + \dots + 2\zeta^{(k+1)(2k+1)}.
\end{align}
\end{proposition}
\begin{proof}
Both the coefficient and the exponent of the highest-order term follow from \eqref{eqn:SigmaLargeZeta}. To find the coefficient of the lowest-order term, we compute the sums of components for $\zeta=0$:
\begin{align}
  \left.\Sigma_{2k}\right|_{\zeta=0} &= 2^{k+1}3^k \left.H_{2k}\right|_{\zeta=0}\left.H_{2(k+1)}(J_2,J_3)\right|_{\zeta=0},\\
  \left.\bar \Sigma_{2k+1}\right|_{\zeta=0} &= 2^{k+1}3^{k+1} \left.H_{2(k+1)}(J_2)\right|_{\zeta=0}\left.H_{2(k+1)}(J_3)\right|_{\zeta=0}.
\end{align}
To evaluate these expressions, we need \eqref{eqn:TrigLimitH2kH2kJ2} and \eqref{eqn:CombAVN8}, as well as \cite{zinnjustin:13}
\begin{equation}
   \left.H_{2k}(J_2,J_3)\right|_{\zeta=0} = 2^{-k} A_{\textrm{UU}}^{(2)}(4n;1,1,1),
\end{equation}
where the numbers
\begin{equation}
   A_{\textrm{UU}}^{(2)}(4n;1,1,1) = 2^{2n}\prod_{i=1}^{n}\frac{(6i-1)(6i-3)!}{(2(n+i))!}
\end{equation}
appear in the enumeration of alternating sign matrices with two $U$-turn boundaries \cite{kuperberg:02}. The final expression for the lowest-order coefficient is a result of the combinatorial identities
\begin{align}
 A_{\textrm{DAD}}(4k+1) &= 3^k A_{\textrm{V}}(2k+1)A_{\textrm{UU}}^{(2)}(4k;1,1,1),\\
 A_{\textrm{DAD}}(4k+3) &= 3^{k+1} N_{8}(2(k+1))A(2k+1)/A_{\textrm{V}}(2k+1).
\end{align}
They follow, for example, from a factorisation of Schur functions into symplectic and orthogonal characters discussed in \cite{behrend:17}.
\end{proof}
We note that the lowest-order term also follows from a rigorous result on the six-vertex model \cite{morin:20}.

\subsection{Discussion}
\label{sec:Discussion}

In this section, we consider the eigenvalue problem 
\begin{equation}
  \label{eqn:HPProblem}
  H|\phi\rangle = E_0 |\phi\rangle, \quad P|\phi\rangle = |\phi\rangle, \quad |\phi\rangle \in V^{2n+1},
\end{equation}
for the XYZ Hamiltonian \eqref{eqn:HamXYZ}. Its space of solutions is one-dimensional \cite{hagendorf:18}. We work with the solution
\begin{equation}
  |\phi_n\rangle = \frac{1}{2}\left(|\psi_n\rangle + |\bar \psi_n\rangle\right).
\end{equation}
Our goal is to compare our results for $|\phi_n\rangle$ to several conjectures by Bazhanov and Mangazeev \cite{mangazeev:10}, and Razumov and Stroganov \cite{razumov:10}, on their solutions, which
we denote by $|\phi_n^{\mathrm{BM}}\rangle$ and $|\phi_n^{\mathrm{RS}}\rangle$, respectively. 
We refer to these conjectures as BM-Conjectures and RS-Conjectures. The comparison to our results suggests that $|\phi_n^{\mathrm{BM}}\rangle =|\phi_n^{\mathrm{RS}}\rangle=|\phi_n\rangle$, but this conclusion remains non-rigorous (even if we admit \cref{conj:PZJ}). The main reason is that the normalisation conventions for $|\phi_n^{\mathrm{BM}}\rangle$ and $|\phi_n^{\mathrm{RS}}\rangle$ differ from the one for $|\phi_n\rangle$. We recall the normalisation of $|\phi_n\rangle$ is implicitly fixed by \eqref{eqn:ChoiceForC} and \eqref{eqn:Psi0}, and by the homogeneous limit \eqref{eqn:HomLimitPsi} and \eqref{eqn:DefN}. In the case of Bazhanov and Mangazeev's work, another reason is that the exact link between the families of polynomials they use and ours is still lacking.

\subsubsection{The BM-Conjectures}

The vector $|\phi_n^{\mathrm{BM}}\rangle$ is normalised so that its components are polynomials in $\zeta$ without a common polynomial factor. This fixes the vector up to an overall numerical factor. Its value is set by the additional requirement
\begin{equation}
 (\phi^{\mathrm{BM}}_{2k})_{\underset{2k}{\underbrace{\scriptstyle{\uparrow\cdots \uparrow}}}\underset{2k+1}{\underbrace{\scriptstyle{\downarrow\cdots \downarrow}}}}\Bigr|_{\zeta=0}=1, \quad (\phi^{\mathrm{BM}}_{2k+1})_{\underset{2k+2}{\underbrace{\scriptstyle{\uparrow\cdots \uparrow}}}\underset{2k+1}{\underbrace{\scriptstyle{\downarrow\cdots \downarrow}}}}\Bigr|_{\zeta=0}=1.
\end{equation}
Bazhanov and Mangazeev formulate their conjectures on $|\phi_n^{\mathrm{BM}}\rangle$ in terms of two families of polynomials $s_n(z),\bar s_n(z)$, where $n\in \mathbb Z$, with integer coefficients. These polynomials are solutions to a difference-differential equation \cite{mangazeev:10}. The polynomials $s_n(z)$ are conjectured to possess the factorisation properties
\begin{subequations}
\label{eqn:FactorSn}
\begin{align}
  s_{2k+1}(y^2) &= \bar c_{2k+1} p_k(y)p_k(-y),\\
  s_{2k}(y^2) &= \bar c_{2k} \bar p_{k+1}(y)q_{k-1}(y), 
\end{align}
\end{subequations}
for each $k\in \mathbb Z$. Here, $p_k(y),q_k(y)$ are polynomials with integer coefficients, and $\bar c_n$ are known constants. Moreover, we use the abbreviation
\begin{equation}
  \bar p_k(y) =  \left(\frac{1+3y}{2}\right)^{k(k-1)}p_{-k}\left(\frac{y-1}{1+3y}\right).
\end{equation}
In his investigations of BM-Conjecture 1, Zinn-Justin \cite{zinnjustin:13} observed for small $k$ the relations
\begin{subequations}
\label{eqn:H2kBMPoly}
\begin{equation}
 H_{2k} = \zeta^{k(k-1)} q_{k-1}(\zeta^{-1}), \quad
  2^{k-1}H_{2k}(J_2,J_3,J_4) = \zeta^{k(k-1)}q_{-k}(\zeta^{-1}),
\end{equation}
and
\begin{align}
  H_{2k}(J_2) &= \zeta^{k(k-1)} \bar p_{-k+1}\left(\zeta^{-1}\right),\quad H_{2k}(J_{3/4}) = \zeta^{k(k-1)}  p_{k-1}\left(\pm\zeta^{-1}\right),\\
  2^{k-1}H_{2k}(J_2,J_{3/4}) &= \zeta^{k(k-1)} p_{-k}\left(\mp\zeta^{-1}\right),\quad 2^{k-1}H_{2k}(J_3,J_4) =\zeta^{k(k-1)}\bar p_k(\zeta^{-1}).
\end{align}
\end{subequations}
We assume that they hold for arbitrary $k$, and use them to
 discuss BM-Conjecture 2, 3 and 4.

BM-Conjecture 3 provides the explicit expression
\begin{equation}
  (\phi_n^{\mathrm{BM}})_{\downarrow\cdots\downarrow} = \zeta^{n(n+1)/2}s_{n}(\zeta^{-2})
\end{equation}
for the component labelled by a polarised spin configuration. Using \cref{prop:APComponent}, we obtain the same component for the vector $|\phi_n\rangle$. For each $k\geqslant 0$, we have
\begin{equation}
  (\phi_{2k})_{\downarrow\cdots\downarrow} = \zeta^k H_{2k}H_{2(k+1)}(J_3,J_4), \quad (\phi_{2k+1})_{\downarrow\cdots\downarrow}=\zeta^{k+1} H_{2(k+1)}(J_3)H_{2(k+1)}(J_4).
\end{equation}
 From the factorisation properties \eqref{eqn:FactorSn}, the relations \eqref{eqn:H2kBMPoly}, and $\bar c_{2k+1} = 1, \,\bar c_{2k}=2^{-k}$ for $k\geqslant 0$ \cite{mangazeev:10}, we conclude that
\begin{equation}
  \label{eqn:EqPolarised}
  (\phi_n^{\mathrm{BM}})_{\downarrow\cdots\downarrow} =  (\phi_n)_{\downarrow\cdots\downarrow}.
\end{equation}
Likewise, BM-Conjecture 4 gives an explicit expression for the component labelled by an alternating spin configuration: For each $k\geqslant 0$,
\begin{subequations}
\label{eqn:BMAltComp}
\begin{align}
  (\phi_{2k}^{\mathrm{BM}})_{\uparrow\downarrow\cdots\uparrow\downarrow\downarrow} &= 2^{k}\zeta^{2k(k-1)}\bar p_{-(k-1)}(\zeta^{-1})q_{k-1}(\zeta^{-1}),\\
  (\phi_{2k+1}^{\mathrm{BM}})_{\uparrow\uparrow\downarrow\cdots\uparrow\downarrow} &= 2^{k}\zeta^{2k^2}\bar p_{-k}(\zeta^{-1})q_{k-1}(\zeta^{-1}).
\end{align}
\end{subequations}
Similarly, we infer from \cref{prop:AlternatingComponent} (and from the invariance of $|\phi_n\rangle$ under translations \cite{hagendorf:18}) the components
\begin{equation}
  (\phi_{2k})_{\uparrow\downarrow\cdots\uparrow\downarrow\downarrow} = 2^{k}H_{2k}H_{2k}(J_2),\quad 
  (\phi_{2k+1})_{\uparrow\uparrow\downarrow\cdots\uparrow\downarrow} = 2^{k}H_{2k}H_{2(k+1)}(J_2).
\end{equation}
We use \eqref{eqn:H2kBMPoly} to compare these expressions to \eqref{eqn:BMAltComp}, and find
\begin{equation}
  \label{eqn:EqAlternating}
  (\phi_{2k}^{\mathrm{BM}})_{\uparrow\downarrow\cdots\uparrow\downarrow\downarrow} = (\phi_{2k})_{\uparrow\downarrow\cdots\uparrow\downarrow\downarrow},\quad 
  (\phi_{2k+1}^{\mathrm{BM}})_{\uparrow\uparrow\downarrow\cdots\uparrow\downarrow} =  (\phi_{2k+1})_{\uparrow\uparrow\downarrow\cdots\uparrow\downarrow}.
\end{equation}
Since the space of solutions to the eigenvalue problem \eqref{eqn:HPProblem} is one-dimensional, the equalities \eqref{eqn:EqPolarised} and \eqref{eqn:EqAlternating} suggest that $|\phi_n^{\mathrm{BM}}\rangle = |\phi_n\rangle$. A proof of this equality would imply that all the components of $|\phi_n\rangle$ are polynomials in $\zeta$. This proof is, however, beyond the scope of this article. Nonetheless, it is interesting to explore the consequences of this equality. In BM-Conjecture $2$, Bazhanov and Mangazeev claim that
\begin{equation}
  \label{eqn:BMAPComp}
  \left(\phi^{\mathrm{BM}}_{n}\right)_{\uparrow\cdots\uparrow\downarrow}= \frac{1}{2n+1} \zeta^{n(n-1)/2}\bar s_n(\zeta^{-2}).
\end{equation}
We obtain the corresponding component of $|\phi_n\rangle$ from \cref{prop:APComponent}. For $k\geqslant 0$, we find
\begin{align}
  (\phi_{2k})_{\uparrow\cdots \uparrow\downarrow} &= \frac{1}{2}\zeta^k\left((1+\zeta)H_{2k}(J_3)H_{2(k+1)}(J_4)+(1-\zeta)H_{2k}(J_4)H_{2(k+1)}(J_3)\right),\\
   (\phi_{2k+1})_{\uparrow\cdots \uparrow\downarrow} &= \frac{1}{2}\zeta^k\left((1-\zeta^2)H_{2(k+1)}H_{2(k+1)}(J_3,J_4)+H_{2k}H_{2(k+2)}(J_3,J_4)\right),
\end{align}
respectively. We match these expressions to \eqref{eqn:BMAPComp} and find, for $k\geqslant 0$,
\begin{align}
  \bar s_{2k}(y^2) &=\frac{(4k+1)}{2y^{2k+1}}((1+y)p_{k-1}(y)p_k(-y)-(1-y)p_{k-1}(-y)p_k(y)),\\
  \bar s_{2k+1}(y^2)&=\frac{(4k+3)}{(2y^{2})^{k+1}}\left((y^2-1)q_{k}(y)\bar p_{k+1}\left(y\right)+\frac12 q_{k-1}(y)\bar p_{k+2}\left(y\right)\right).
\end{align}
These relations resemble the factorisations \eqref{eqn:FactorSn} for $s_n(y^2)$. To our best knowledge, they have not been reported in the literature.

\subsubsection{The RS-Conjectures}

The vector $|\phi_n^{\mathrm{RS}}\rangle$ is normalised so that its components are polynomials in $\zeta$ without a common polynomial factor, too.\footnote{Razumov and Stroganov do not explicitly mention the absence of a common polynomial factor, but this assumption appears to be implicit in their work.} Razumov and Stroganov argue that this normalisation convention implies that the component of highest degree is $(\phi_n^{\mathrm{RS}})_{\downarrow\cdots\downarrow}$.
 According to RS-Conjecture 4.2, its highest-order term is
\begin{equation}
  (\phi_n^{\mathrm{RS}})_{\downarrow\cdots\downarrow} = a_n\zeta^{n(n+1)/2}+\cdots,
\end{equation}
where $a_n$ is non-zero and $\cdots$ denotes lower-order terms. 
Razumov and Stroganov consider $a_n=1$, which fixes the remaining overall numerical factor in their normalisation. By RS-Conjecture 4.4, this choice implies
\begin{align}
  (\phi_{2k}^{\mathrm{RS}})_{\downarrow\cdots\downarrow}  
  &= A_{\mathrm V}(2k+1)^2\zeta^k + \cdots + \zeta^{k(2k+1)},\\ (\phi_{2k+1}^{\mathrm{RS}})_{\downarrow\cdots\downarrow}  
  &= N_{8}(2k+2)^2\zeta^{k+1} + \cdots + \zeta^{(k+1)(2k+1)}.
\end{align}
Using \cref{prop:APComponentComb}, we find exactly the same expressions for the components $(\phi_{2k})_{\downarrow\cdots\downarrow}$ and $(\phi_{2k+1})_{\downarrow\cdots\downarrow}$, which suggests (but does not prove) the equality $|\phi_n^{\mathrm {RS}}\rangle = |\phi_n\rangle$. 
Another property supporting this equality is the sum rule
\begin{equation}
  \sum_{\bm \alpha} (\phi_n^{\mathrm{RS}}(\zeta))_{\bm \alpha} = 2^{-n(n-1)/2}(\zeta-1)^{n(n+1)/2}(\phi_n^{\mathrm{RS}}(\zeta'))_{\downarrow\cdots\downarrow},
\end{equation}
where $\zeta'=(\zeta+3)/(\zeta-1)$, which is formulated in RS-Conjecture 5.2. It, indeed, holds for $|\phi_n\rangle$ as well, which straightforwardly follows from \cref{prop:SumOfComps} and its proof.

\section{Conclusion}
\label{sec:Conclusion}

In this article, we have investigated several scalar products involving a particular eigenvector of the transfer matrix of the inhomogeneous supersymmetric eight-vertex model with periodic boundary conditions. We have found explicit expressions for them in terms of the elliptic Tsuchiya determinant. 
In the homogeneous limit, they allowed us to compute the scalar products $S_n$ and $\bar S_n^\pm$ involving the basis vectors $|\psi_n\rangle,\,|\bar \psi_n\rangle$ of the eigenspace for the transfer-matrix eigenvalue $\Theta_n$ of the homogeneous eight-vertex model.
Our main results are \cref{thm:S,thm:SBar}. They provide new and explicit expressions for these scalar products in terms of special polynomials introduced by Zinn-Justin and Rosengren. We have used them to compute several components of the basis vectors, as well as the sum
of their components, establish their polynomial nature and compute their trigonometric limit. 

We now discuss several open problems and generalisations of the present work. First, our results assume that \cref{conj:PZJ} holds. The proof of this conjecture remains a challenge, which would undoubtedly provide even more insight into the properties of the eigenvector $|\Psi_n\rangle$ and its homogeneous limit. Second, the comparison of our results to the investigations of Bazhanov and Mangazeev \cite{mangazeev:10}, and Razumov and Stroganov \cite{razumov:10}, strongly suggests that all the components of $|\psi_n\rangle$ are polynomials in $\zeta$ with integer coefficients. Proving the polynomial nature in $\zeta$ would be a step forward to settling most (if not all) of their conjectures.
Third, we observe for small $n$ that the integer coefficients of the powers of $\zeta$ in a given component of $|\psi_n\rangle$ all have the same sign. This observation suggests that they could have a combinatorial meaning. We note that Hietala has recently found a combinatorial interpretation of the polynomial $H_{2k}$ in terms of a partition function of the three-colour model \cite{hietala:20}. 
Similar results for $H_{2k}(J_i),H_{2k}(J_i,J_j),\dots$ remain to be found.
 Fourth, the results of this article suggest that an exact finite-size computation of the emptiness or boundary emptiness formation probability for the supersymmetric eight-vertex model could be possible. In the trigonometric limit, these correlation functions are known \cite{cantini:12,morin:20}. 
Finally, we mention that there is a conjecture for a simple eigenvalue of the transfer matrix of the inhomogeneous supersymmetric eight-vertex model with open boundary conditions \cite{hagendorf:20}.
A characterisation of its eigenspace, similar to \cref{conj:PZJ}, is still to be found.

\subsubsection{Acknowledgements}
This work was supported by the Fonds de la Recherche Scientifique-FNRS and the Fonds Wetenschappelijk Onderzoek-Vlaanderen (FWO) through the Belgian Excellence of Science (EOS) project no. 30889451 ``PRIMA -- Partners in Research on Integrable Models and Applications''. SB is supported by the FNRS aspirant fellowship FC33665. We thank
Jules Lamers, Jean Li\'enardy and Hjalmar Rosengren for discussions. Furthermore, CH thanks the Laboratoire de Physique Th\'eorique des Mod\`eles Statistiques, Orsay, France, where part of this work was done, for hospitality.

\end{document}